\documentclass[preprint,12pt,authoryear]{elsarticle}
\usepackage{fullpage}
\usepackage{times}

\usepackage{lineno}
\usepackage[colorlinks]{hyperref}
\usepackage[T1]{fontenc} 
\usepackage{amsthm,amsmath,amsfonts,amssymb}
\usepackage{bm}
\usepackage{textcomp}
\usepackage{graphicx}
\usepackage{subfigure}
\usepackage{float}
\usepackage{algorithm}  
\usepackage{algorithmic}
\usepackage{url}

\usepackage[utf8]{inputenc}
\usepackage{amsmath, amsfonts, amssymb, amsthm}
\usepackage{geometry}
\usepackage{enumitem}
\usepackage{comment}

\usepackage[round]{natbib}
\usepackage{hyperref}
\hypersetup{colorlinks=true, citecolor=blue, linkcolor=black}

\usepackage{listings}
\usepackage{xcolor} 

\usepackage{titlesec}
\titleformat{\paragraph}[block]{\bfseries}{\theparagraph}{1em}{}

\newtheorem{theorem}{Theorem}
\newtheorem{lemma}{Lemma}
\newtheorem{proposition}{Proposition}
\newtheorem{definition}{Definition}

\journal{Computational Statistics and Data Analysis}

\begin{document}

\begin{frontmatter}

\title{Sobolev-Regularized Objective Functions for \\ Robust Pairwise Alignment of Functional Data} 
\author{Wei Wu}
\address{Department of Statistics, Florida State University, \\ 117 N. Woodward Ave., Tallahassee, FL 32306, USA}

\date{}

\begin{abstract}
Functional data registration is a critical challenge in modern statistics, essential for separating phase variability from amplitude variability. While derivative-based frameworks offer mathematically elegant solutions, their dependence on signal velocities renders them susceptible to additive noise. This study proposes and evaluates a family of robust, Sobolev-regularized objective functions for the pairwise alignment of functional data, operating entirely within the original function space to avoid the need for numerical differentiation of the data. We define our optimization over a second-order Sobolev space and utilize the Centered Log-Ratio (CLR) transform to represent the warping functions. By penalizing both the velocity and acceleration of the centered log-derivative, this geometric approach preempts degenerate ``pinching'' artifacts and ensures the resulting warps are strictly monotonic, valid diffeomorphisms. In practice, this allows for highly efficient, unconstrained optimization within a finite-dimensional space. We systematically investigate four distinct pairwise data mismatch formulations: a Standard $\mathbb{L}^2$ baseline, a Symmetric $\mathbb{L}^2$ formulation, an Isometry ($\mathbb{L}^2$-preserving) mapping, and a Jacobian-weighted $\mathbb{L}^2$ functional. We establish robust theoretical foundations for these methods, proving the existence of optimal warps and the asymptotic consistency of the finite-dimensional estimators. Our results demonstrate that this CLR-regularized framework offers a powerful, computationally scalable, and noise-robust alternative to traditional derivative-based registration.


\end{abstract}

\end{frontmatter}

\section{Introduction}
\label{sec:introduction}

Temporal phase variability remains a central challenge in functional data analysis. In function registration, the goal is typically to isolate this timing variation from the actual signal magnitude using a time warping function. In many studies, it is critically important to find the aligned functions because warping is considered a nuisance variable in the measurement process, and its variability needs to be removed \citep{ramsay2005functional}. In other cases, the phase itself is considered an essential feature of the data \citep{marron2015functional}. For either purpose, one must estimate optimal time warpings to properly align functional observations. A common space of valid time warping functions is defined by the group of orientation-preserving diffeomorphisms:
$$
\Gamma_c = \{ \gamma : [0,1] \to [0,1] \mid \gamma(0) = 0, \gamma(1) = 1, 0 < \gamma'(t) < \infty, \forall t \in [0,1] \} 
$$
where $\gamma'$ denotes the derivative of $\gamma$. Over the past two to three decades, various approaches have been developed for robust and efficient estimation within this space. Early approaches formulated a least-squares problem by representing the warping function, or a mathematical transformation thereof, via a linear combination of polynomial or B-spline basis functions, obtaining the warping by estimating the corresponding coefficients \citep{ramsay1998curve, gervini2004self, eilers2004parametric, james2007curve, telesca2008bayesian}. 

However, the space $\Gamma_c$ is strictly nonlinear under the conventional $\mathbb{L}^2$ metric. When optimizing standard distances without proper geometric constraints, these early frameworks suffered from two primary flaws. First, they are inherently \textit{asymmetric}: the warping required to align $f \to g$ does not yield the same cost as the inverse alignment $g \to f$. Second, they are highly susceptible to the \textit{``pinching effect''}. In this degenerate scenario, the warping velocity $\gamma'$ approaches zero or spikes toward infinity, excessively compressing or expanding the temporal domain simply to match signal amplitudes, thereby destroying the true underlying phase structure.  While conventional penalty methods attempted to regularize the warp by penalizing simple derivatives (e.g., $\gamma''$), these approaches fundamentally ignored the intrinsic, non-Euclidean geometry of the warping space.

To rigorously address these geometric challenges, recent breakthroughs conducted registration by minimizing the Fisher-Rao metric via the Square-Root Velocity Function (SRVF) \citep{srivastava2011}. By transforming the functional data into SRVFs, the complicated Fisher-Rao metric is flattened into a standard $\mathbb{L}^2$ distance. This framework is mathematically elegant, ensuring symmetry and providing an isometry-invariant mapping where registration can be performed reliably. It incorporates the warping derivative directly into the transformed signal representation and naturally preempts the pinching effect. This geometric formulation has spurred highly successful applications in analyzing phase variations for functional regression \citep{ahn2020regression}, classification \citep{tucker2013generative}, and functional principal component analysis \citep{lee2016fpca}.

In parallel to deterministic optimization techniques, SRVF-based Bayesian methods have emerged as a powerful paradigm for functional data registration, exploring more comprehensive possibilities regarding time warping, particularly the quantification of alignment uncertainty. By treating the warping function as a random element, these approaches place flexible priors over the space of valid diffeomorphisms to generate posterior distributions of the temporal warps
\citep{cheng2016bayesian, kurtek2017geometric}. Further advancements have modeled the warping functions utilizing Gaussian processes \citep{lu2017bayesian} or point processes \citep{Bharath2020landmark} and extended these principles to handle highly challenging scenarios, such as the simultaneous registration of noisy, sparse, and fragmented functional data \citep{matuk2021bayesian}, or navigating complex posterior spaces using Hamiltonian Monte Carlo \citep{tucker2021multimodal}. Because of the non-linearity of $\Gamma_c$, the likelihood functions and prior distributions in these Bayesian models depend heavily on the SRVF representation or operate within the tangent space of the nonlinear SRVF hypersphere, inherently relying on a local linear approximation of the manifold.

Despite their geometric elegance, SRVF-based methods, as well as the probabilistic models that rely on them, share a fundamental limitation: the transformation requires the temporal derivative of the observed signals. In practical scenarios where data are corrupted by additive noise, numerical differentiation drastically amplifies high-frequency fluctuations. This instability means these methods typically require pre-smoothing, which may blur away the underlying structural features required to properly align the curves.
Consequently, this study aims to perform registration directly in the \textit{original function domain}. By avoiding signal derivatives, we maintain robustness against noise while exploring a family of mismatch terms designed to be symmetric with respect to the input signals, ensuring the registration is not biased by the directional choice of a template.

In this study, we introduce a new, Sobolev-regularized deterministic framework for pairwise functional registration. While extending alignment to \textit{multiple functions} is a natural next step, restricting our current focus to the pairwise setting (aligning a source function $f$ to a target template $g$) allows us to rigorously isolate and establish the theoretical properties, symmetries, and structural biases of the objective functions. Our framework presents two primary contributions:
\begin{enumerate}
    \item \textit{Sobolev-Regularized CLR Penalty:} Rather than utilizing conventional nonlinear approximations, we adopt the Centered Log-Ratio (CLR) transform \citep{ma2024} to map the strictly constrained warping manifold into an unconstrained, linear Hilbert subspace. While recent Bayesian frameworks \citep{ma2025new} have utilized this transform under a standard $\mathbb{L}^2$ structure, we propose to govern the optimization using a strictly defined zero-mean Sobolev $\mathbb{H}$-norm penalty. By penalizing both the velocity and acceleration of the centered log-derivative, this higher-order regularization rigorously respects the intrinsic geometry of time warping and effectively eliminates the pinching effect, ensuring globally smooth diffeomorphisms. Furthermore, operating within this unconstrained linear space enables highly efficient gradient-based optimization.
    
    \item \textit{Symmetric Data Mismatch Functionals:} To fully leverage our derivative-free approach, we propose and investigate four distinct mismatch functionals formulated entirely within the original function space. In addition to the standard $\mathbb{L}^2$ baseline, we introduce three symmetric alternatives: a Symmetric $\mathbb{L}^2$ formulation, an Isometry ($\mathbb{L}^2$-preserving) mapping, and a Jacobian-weighted formulation. This diversity allows us to explicitly address the trade-offs between mathematical symmetry and signal integrity.
\end{enumerate}
We carefully examine this regularized framework across theory, methodology, and numerical computations. By establishing the asymptotic consistency of our estimators, this paper provides a comprehensive, computationally efficient solution for robust phase-amplitude decoupling in the presence of additive noise.

The remainder of this paper is organized as follows. Section 2 defines the mathematical framework for the warping manifold and the second-order Sobolev space $\mathbb{H}$. Section 3 details the baseline and three symmetric objective functions, providing the theoretical justification for the amplitude-preserving and geometric formulations, and establishing guarantees for the existence of optimal warps. Section 4 describes the finite-basis estimation strategy and the unified gradient descent algorithm, including a noise-free asymptotic theory and a complexity analysis. We also demonstrate the performance of these methods through comprehensive simulation studies and an application to a real-world acoustic dataset. Finally, Section 5 provides a concluding discussion. All technical proofs and mathematical details are deferred to the Appendices.

\section{Manifold Linearization and Sobolev Optimization Spaces}
\label{sec:space}

The robustness of the proposed registration framework is predicated on the geometric structure of the warping domain. A significant hurdle in functional registration is that the space of warping functions is not naturally a vector space; standard pointwise addition or scalar multiplication of two diffeomorphisms does not necessarily yield another valid diffeomorphism. To address this, we adopt the manifold linearization approach to transition from the nonlinear manifold to a structured vector space.

\subsection{Review: The Manifold of Warping Functions as a Vector Space}

Following the framework established by \citet{ma2024}, we briefly review the mathematical structure of the warping manifold that serves as the foundation for our optimization. 
Let $\Gamma \subset \Gamma_c$ denote the space of time warping functions considered in this study, defined as the set of strictly increasing, absolutely continuous diffeomorphisms with bounded derivatives on the unit interval $I = [0, 1]$:
\begin{equation}
\Gamma = \{ \gamma : [0,1] \to [0,1] \mid \gamma(0) = 0, \gamma(1) = 1, 0 < m_\gamma < \gamma'(t) < M_\gamma < \infty, \forall t \in [0,1] \}
\end{equation}
The strictly positive lower bound $m_\gamma$ and finite upper bound $M_\gamma$ are mathematically essential. Without them, the integration of the logarithm of the derivative, $\log \gamma'(t)$, can diverge. These bounds guarantee that the CLR transform remains strictly well-defined, thereby establishing a rigorous isometric isomorphism between the constrained manifold $\Gamma$ and an unconstrained $\mathbb{L}^2$ subspace.

While the set of transformations $\Gamma$ forms a \textit{group} under composition, it lacks a natural linear structure, which complicates direct optimization. To address this, an isometric isomorphism $\Phi$ is employed to map these constrained functions into a bounded, zero-mean $\mathbb L^2$ subspace. This linearization allows for the application of standard linear operations; the resulting estimates are then mapped back to the diffeomorphism group $\Gamma$ to perform the signal alignment.

Let $\mathbb{L}_{0, \infty}(I)$ denote a specific subspace of $\mathbb{L}^2(I)$, consisting of zero-mean, square-integrable functions with a finite $\mathbb{L}^\infty$ norm:
\begin{equation}
    \mathbb{L}_{0, \infty}(I) = \left\{ \psi \in \mathbb{L}^2(I) \;\middle|\; \int_0^1 \psi(t) \, dt = 0, \text{ and } \|\psi\|_\infty < \infty \right\}.
\end{equation}
The mapping $\Phi: \Gamma \to \mathbb{L}_{0, \infty}(I)$, defined by the Centered Log-Ratio (CLR) transform, provides the required isometric isomorphism \citep{egozcue2006hilbert}:
\begin{equation}
    \psi(t) = \Phi(\gamma)(t) = \log \gamma'(t) - \int_0^1 \log \gamma'(s) \, ds.
    \label{eq:CLR}
\end{equation}
In this formulation, $\log \gamma'(t)$ captures the local log-derivative of the warp, while the integral term serves as a centering constraint. Its inverse, $\Phi^{-1}$, is the normalized exponential map:
\begin{equation}
    \gamma(t) = \Phi^{-1}(\psi)(t) = \frac{\int_0^t e^{\psi(s)} \, ds}{\int_0^1 e^{\psi(s)} \, ds}.
    \label{eq:invCLR}
\end{equation}
The denominator acts as a normalizing constant ensuring that $\gamma(1) = 1$ regardless of the magnitude of $\psi$, effectively mapping the entire vector space $\mathbb{L}_{0, \infty}(I)$ back onto the manifold $\Gamma$.

A significant advantage of this characterization is that $\Gamma$ inherits well-defined linear and inner-product operations via the pullback of $\Phi$. Specifically, for any $\gamma_1, \gamma_2 \in \Gamma$ and $a \in \mathbb{R}$, the vector space operations are defined as follows:
\begin{itemize}
    \item \textbf{Addition}: $\gamma_1 \oplus \gamma_2 = \Phi^{-1} \left( \Phi(\gamma_1) + \Phi(\gamma_2) \right)$
    \item \textbf{Scalar Multiplication}: $a \odot \gamma_1 = \Phi^{-1} \left( a \Phi(\gamma_1) \right)$
    \item \textbf{Inner-Product}: $\langle \gamma_1, \gamma_2 \rangle_\Gamma = \langle \Phi(\gamma_1), \Phi(\gamma_2) \rangle_{\mathbb{L}^2}$
\end{itemize}

Through these definitions, the traditionally nonlinear time-warping manifold is transformed into a linear, inner-product space. This shift makes the underlying concepts far more intuitive and allows for the use of conventional linear algebraic operations. Furthermore, because $\Phi$ is an isometry under the metric induced by the $\mathbb{L}^2$ inner product on $\mathbb{L}_{0, \infty}(I)$, the optimization of warping functions can be conducted in a flat, conventional $\mathbb{L}^2$ subspace. This facilitates unconstrained optimization while implicitly and automatically maintaining the necessary diffeomorphism constraints, such as monotonicity and fixed boundary conditions, on $\Gamma$.

\subsection{The Sobolev Optimization Space}

While the linearized space $\mathbb{L}_{0, \infty}(I)$ provides the necessary vector structure for the warping functions, it is too broad for stable optimization as it includes functions with insufficient regularity. To ensure that our estimated warping functions are smooth and physically plausible, we define the optimization space $\mathbb{H}$ as a second-order Sobolev space. This space acts as a dense, smooth subspace of $\mathbb{L}_{0, \infty}(I)$, providing the essential functional architecture for both finite-basis expansion and Tikhonov-style regularization \citep{deBoor2001, adams2003sobolev}.

\subsubsection{The Second-order Sobolev Space}

\begin{definition}[The Sobolev Space $\mathbb{H}$]
We define $\mathbb{H}$ as a Sobolev space of order $p=2$, restricted to the zero-mean subspace:
\begin{equation*}
    \mathbb{H} = \left\{ \psi \in C^1(I) \mid \psi' \text{ is absolutely continuous}, \psi'' \in \mathbb{L}^2(I), \text{ and } \int_0^1 \psi(t)dt = 0 \right\}
\end{equation*}
$\mathbb{H}$ is equipped with the inner-product:
\begin{equation*}
    \langle \psi_1, \psi_2 \rangle_{\mathbb{H}} = \int_0^1 \psi_1'(t) \psi_2'(t) dt + \int_0^1 \psi_1''(t) \psi_2''(t) dt
\end{equation*}
The induced norm is strictly defined as $\|\psi\|_{\mathbb{H}} = \sqrt{\langle \psi, \psi \rangle_{\mathbb{H}}}$.
\end{definition}

Notice that our definition of the inner product and induced norm intentionally omits the zeroth-order term $\|\psi\|_{\mathbb{L}^2}^2$ found in the standard $W^{2,2}$ Sobolev norm. This omission is mathematically rigorous due to the zero-mean constraint $\int_0^1 \psi(t) dt = 0$. By the Poincar\'{e} inequality, the overall magnitude of any zero-mean function is strictly bounded by the magnitude of its first derivative (i.e., $\|\psi\|_{\mathbb{L}^2}^2 \leq C \|\psi'\|_{\mathbb{L}^2}^2$). Consequently, the standard Sobolev norm and our reduced norm are strictly equivalent and define the exact same topological space. Dropping this redundant term not only streamlines the computational optimization, but also ensures the penalty strictly targets the geometric shape of the deformation.

First, we establish in Proposition \ref{prop:hilbert} that $\mathbb{H}$ is indeed a Hilbert space. Proving this completeness property is critical, as it serves as the foundational requirement for both our theoretical and numerical frameworks. Mathematically, the Hilbert structure ensures the well-posedness of the optimization problem, providing the necessary topological properties to guarantee that a minimizer exists and that our subsequent consistency results hold. Practically, this allows us to fully exploit basis-expansion methods within a rigorous Sobolev framework, ensuring that finite-dimensional approximations remain robust and convergent. By mapping the nonlinear manifold of warping functions into this complete, linear Hilbert space, we transform a constrained registration problem into a manageable and theoretically sound unconstrained optimization.

\begin{proposition}[Hilbert Space Property]
\label{prop:hilbert}
The Sobolev space $\mathbb{H}$, equipped with the inner product $\langle \cdot, \cdot \rangle_{\mathbb{H}}$, is a Hilbert space.
\end{proposition}

The proof of Proposition \ref{prop:hilbert} is given in \ref{app:hilbert_proof}.  
Functions within $\mathbb{H}$ are required to be continuously differentiable on the compact interval $I$. In Proposition \ref{prop:bound}, we provide an explicit, constructive proof that bounding the $\mathbb{H}$-norm tightly bounds the $\mathbb{L}^\infty$ norms of both $\psi$ and its derivative. This direct mathematical relationship is fundamental to our framework: it implies that by controlling the $\mathbb{H}$-norm during optimization, we explicitly constrain the maximum amplitude of the log-derivative and its rate of change. This ensures that the resulting warping function $\gamma = \Phi^{-1}(\psi)$ remains a valid, smooth diffeomorphism with derivatives strictly bounded away from zero and infinity, effectively preventing the ``pinching'' artifacts common in unregularized registration.

\begin{proposition}[Uniform Stability]
\label{prop:bound}
For any $\psi \in \mathbb{H}$, the $\mathbb{L}^\infty$ norms of $\psi$ and its derivative $\psi'$ are strictly bounded by the $\mathbb{H}$-norm:
$$
    \|\psi\|_\infty \leq \|\psi'\|_{\mathbb{L}^2} \leq \|\psi\|_{\mathbb{H}}
$$
$$
    \|\psi'\|_\infty \leq \sqrt{2} \|\psi\|_{\mathbb{H}}
$$
\end{proposition}

See the proof of this proposition in \ref{app:bound_proof}. We emphasize that these two uniform bounds serve distinct and critical roles in ensuring the geometric validity of the resulting warping function $\gamma$. Because $\psi$ acts as the CLR representation of the deformation, the first bound on $\|\psi\|_\infty$ explicitly guarantees that the slope of the warp, proportional to $\exp(\psi(t))$, is strictly bounded away from zero and infinity. This physically prevents the transformation from exhibiting ``pinching'' (zero or infinite slope), and rigorously confirms the continuous embedding $\mathbb{H} \subset \mathbb{L}_{0, \infty}(I)$. The second bound on $\|\psi'\|_\infty$ restricts the maximum rate of change of the CLR representation, which corresponds to the relative acceleration (or curvature) $\gamma'' / \gamma'$ of the warp. By controlling this derivative, we ensure that the deformation remains globally smooth and free from sharp, high-frequency oscillations \citep{ramsay1998curve}.

\subsubsection{The Sobolev Penalty}

We define the roughness of a warping function through the squared $\mathbb{H}$-norm of its centered log-derivative $\psi = \Phi(\gamma)$. Based on the inner product of our Sobolev space, for any $\psi \in \mathbb{H}$, the regularization penalty $\mathcal{R}(\psi)$ is formulated as:
\begin{equation}
    \mathcal{R}(\psi) = \|\psi\|_{\mathbb{H}}^2 = \int_0^1 \left( \psi'(t) \right)^2 dt + \int_0^1 \left( \psi''(t) \right)^2 dt \label{eq:sobolev_penalty}
\end{equation}
The objective functions for the proposed methods are constructed by incorporating this roughness penalty into the mismatch functionals. By penalizing the first and second derivatives simultaneously, we regularize both the \textit{velocity} and \textit{acceleration} of the warping function in the log-space.

This formulation provides a unified geometric solution to the inherent limitations of previous penalty strategies in functional alignment. Historically, regularization techniques can be categorized by the space they operate in and the specific derivative order they penalize, as summarized in Table \ref{tab:penalty_comparison} (where $\gamma_{id}(t) = t$ is the identity warping). 

\begin{table}[htpb]
\centering
\caption{Comparison of Regularization Penalties in Functional Alignment}
\label{tab:penalty_comparison}
\renewcommand{\arraystretch}{1.3}
\resizebox{\textwidth}{!}{%
\begin{tabular}{llcl}
\hline\hline
\textbf{Reference} & \textbf{Penalty Formulation} & \textbf{Space} & \textbf{Primary Structural Limitation} \\ \hline
\citet{tang2008pairwise} & $\|\gamma - \gamma_{id}\|_{\mathbb{L}^2}^2$ & Original & Biases against large, valid phase variations \\
\citet{james2007curve} & $\|(\gamma')^{-1} - 1\|_{\mathbb{L}^2}^2$ & Original & Requires constrained optimization, division issues \\
\citet{srivastava2016functional} & $\|\sqrt{\gamma'} - 1\|_{\mathbb{L}^2}^2$ & Original & Heavily penalizes smooth, linear stretching \\
\citet{kneip2008combining} & $\|\gamma''\|_{\mathbb{L}^2}^2$ & Original & Fails to inherently prevent ``folding'' \\ \hline
\citet{ma2024} & $\|\psi\|_{\mathbb{L}^2}^2$ & CLR & Biases toward identity; lacks smoothness constraints \\
\citet{ramsay1998curve} & $\|\psi'\|_{\mathbb{L}^2}^2 \quad (\psi' = \frac{\gamma''}{\gamma'})$ & CLR & Geometrically deficient; permits non-differentiable ``kinks'' \\
Theoretical Formulation & $\|\psi''\|_{\mathbb{L}^2}^2$ & CLR & Fails topologically (semi-norm); unpenalized null space \\ \hline
\textbf{Proposed Sobolev} & $\|\psi'\|_{\mathbb{L}^2}^2 + \|\psi''\|_{\mathbb{L}^2}^2$ & CLR ($\mathbb{H}$) & \textit{Establishes Hilbert structure and bounded embedding} \\ \hline\hline
\end{tabular}%
}
\end{table}

The most fundamental approaches define the regularization directly in the original warping space $\Gamma$, utilizing displacement penalties \citep{tang2008pairwise}, velocity deviations \citep{james2007curve}, SRVF-based metrics \citep{srivastava2016functional}, or raw Euclidean acceleration \citep{kneip2008combining}. While these methods have established rigorous frameworks for functional alignment, operating directly in the constrained, nonlinear manifold of diffeomorphisms presents significant theoretical and practical challenges. First, penalties in the original domain frequently target the absolute magnitude of the deformation, such as pushing the displacement toward zero ($\|\gamma - \gamma_{id}\|_{\mathbb{L}^2}^2$), or the velocity toward one ($\|(\gamma')^{-1} - 1\|_{\mathbb{L}^2}^2$ and $\|\sqrt{\gamma'} - 1\|_{\mathbb{L}^2}^2$), or the acceleration toward zero ($\|\gamma''\|_{\mathbb{L}^2}^2$). This inherently biases the estimator toward the identity function, artificially penalizing large, entirely valid phase variations that frequently occur in physical and biological data. Second, from a theoretical perspective, proving the asymptotic consistency of estimators defined directly on $\Gamma$ typically requires imposing strict, artificial boundary assumptions on the target warping functions (e.g., assuming a priori that $0 < \epsilon \leq \gamma'(t) \leq M$ for some arbitrary constants) to prevent the manifold from collapsing to non-differentiable step functions. This same limitation requires computationally expensive constrained optimization techniques to physically prevent domain ``folding'' ($\gamma'(t) \leq 0$) during the estimation process.

Transforming the problem via the Centered Log-Ratio mapping elegantly resolves the manifold constraints and boundary issues of traditional methods. Because the derivative $\gamma'(t) \propto \exp(\psi(t))$ is strictly positive for any $\psi \in \mathbb{H}$, domain folding is mathematically impossible, eliminating the need for expensive constrained optimization. However, effectively regularizing in this unconstrained space requires precise control over the geometric structure of $\psi$; previous CLR approaches fundamentally lack this control, leaving alignments vulnerable to distinct topological failures. For instance, \citet{ma2024} employ a conventional $\mathbb{L}^2$ penalty, $\|\psi\|_{\mathbb{L}^2}^2$. Because $\psi(t) = 0$ corresponds exactly to the identity warp $\gamma_{id}$, this zeroth-order penalty suffers from the exact same fundamental flaw as original-domain methods: it restricts the overall magnitude of the deformation rather than its roughness. Consequently, it artificially biases the estimator toward the identity, unduly penalizing large but mathematically valid time warpings. Furthermore, it lacks structural awareness of the warp’s shape. This allows the optimizer to introduce unpenalized high-frequency noise, yielding jagged derivatives that physically distort the alignment.

Moving beyond zeroth-order penalties, regularizing the derivatives of $\psi$ directly targets the relative acceleration (or relative roughness) of the warp, given by $\psi'(t) = \gamma''(t)/\gamma'(t)$. This exact ratio forms the basis of the classic continuous registration framework introduced by \citet{ramsay1998curve}, which utilizes a first-order penalty, $\|\psi'\|_{\mathbb{L}^2}^2$. While a significant improvement, a purely first-order penalty lacks crucial geometric information regarding higher-order smoothness, permitting the optimal $\psi$ to be piecewise linear. This allows for continuous but non-differentiable ``kinks'' in the log-velocity, translating to abrupt, physically unnatural shifts in the acceleration of the resulting warp. Conversely, applying a purely second-order penalty ($\|\psi''\|_{\mathbb{L}^2}^2$) entirely loses control over the magnitude of the stretch. Because the null space of the second derivative includes all linear functions ($\psi(t) = at + b$), the optimizer could select an arbitrarily steep slope without incurring any penalty. 

More importantly, under the intrinsic CLR zero-mean constraint ($\int_0^1 \psi(t)dt = 0$), relying on only a single derivative fails to provide the comprehensive mathematical foundation required for rigorous theoretical guarantees. A purely first-order penalty constructs a space that, while forming a valid norm, remains geometrically deficient and permits non-differentiable kinks. Conversely, a purely second-order penalty fails topologically altogether: it remains a mere semi-norm with an unpenalized linear null space that inherently prevents the formation of a Hilbert space. By simultaneously restricting both $\|\psi'\|_{\mathbb{L}^2}^2$ and $\|\psi''\|_{\mathbb{L}^2}^2$, our complete Sobolev penalty perfectly balances these competing requirements and establishes the proper Hilbert structure. The first-derivative term strictly bounds the maximum allowable relative stretch (closing the linear null space), while the second-derivative term enforces continuous differentiability to prevent microscopic phase noise and abrupt kinks. Together, they construct the exact topological space $\mathbb{H}$ needed to establish the continuous embedding $\mathbb{H} \subset \mathbb{L}_{0, \infty}(I)$ via Proposition \ref{prop:bound}. This mathematical structure intrinsically ensures the warping derivative remains uniformly bounded away from zero and infinity, replacing the artificial boundary assumptions of traditional methods with a rigorous, geometrically meaningful topological guarantee.

\section{Sobolev-Regularized Alignment Framework}
\label{sec:sobolev}

Building upon the geometric and Sobolev foundations previously established, this section details the construction of the global registration objective. The framework is designed to find a valid diffeomorphic warping that optimizes the balance between robust signal alignment in the original function space and structural smoothness. By utilizing the CLR transformation, we formulate a regularized objective that avoids the optimization challenges of explicit inverse-consistency constraints, enabling highly efficient, unconstrained optimization.

\subsection{The Global Objective and Mismatch Options}
\label{subsec:mismatch}

The goal of the registration framework is to find an optimal centered log-derivative $\hat{\psi}$ that minimizes a composite energy functional. Throughout this study, we restrict our focus to the pairwise registration problem. Let $f, g \in \mathbb{L}^2([0,1])$ represent two observed functional data samples, where $g$ acts as the fixed template and $f$ is the source signal to be warped. We define the composite energy functional as:
\begin{equation}
    O(\psi; f, g) = D(f, g, \Phi^{-1}(\psi)) + \lambda \|\psi\|_{\mathbb{H}}^2
    \label{eq:obj}
\end{equation}
where $D(f, g, \gamma)$ represents the data mismatch term, $\lambda > 0$ is the regularization parameter, and $\|\psi\|_{\mathbb{H}}^2 = \int_0^1 (\psi'(t))^2 dt + \int_0^1 (\psi''(t))^2 dt$ is the Sobolev roughness penalty. The optimal $\hat{\psi}$ is then formally defined as the minimizer of this energy:
\begin{equation*}
    \hat{\psi} = \arg \min_{\psi \in \mathbb{H}} O(\psi; f, g)
\end{equation*}
We systematically consider four specific formulations for the mismatch term $D$ in Equation \eqref{eq:obj}, each offering a different geometric approach to aligning the original signal intensities:

\begin{enumerate}
    \item \textbf{Standard $\mathbb{L}^2$ (Baseline):}  
    For a given pair of signals $f, g \in \mathbb{L}^2(I)$, the classical distance operator is:
    \begin{equation}
        D_1(f, g, \gamma) = \int_{0}^{1} (f(t) - g(\gamma(t)))^2 dt 
    \end{equation}
    This formulation represents the conventional Euclidean distance commonly used in functional data analysis \citep{ramsay2005functional}. It lacks inherent symmetry, as the cost is defined solely in the coordinate system of $f$, making it sensitive to the choice of reference signal.

    \item \textbf{Symmetric $\mathbb{L}^2$:}  
    Drawing from the symmetric framework established in \citet{tagare2009symmetric_theory}, this formulation addresses the inherent bias of the classical asymmetric distance. We achieve exact data bidirectionality by summing the forward and backward residuals:
    \begin{equation}
        \begin{aligned}
            D_2(f, g, \gamma) &= \frac{1}{2}\int_{0}^{1} (f(t) - g(\gamma(t)))^2 dt + \frac{1}{2}\int_{0}^{1} (g(\tau) - f(\gamma^{-1}(\tau)))^2 d\tau \\
            &= \int_{0}^{1} (f(t) - g(\gamma(t)))^2 \left (\frac{1 + \gamma'(t)}{2} \right ) dt 
        \end{aligned} 
    \end{equation}
    The mismatch term $D_2(f, g, \gamma)$ structurally removes directional bias. The integration weight $(\frac{1 + \gamma'(t)}{2})$ --- derived via a change of variables in the second integral --- ensures the energy is invariant to the choice of reference signal, providing strict inverse consistency for the alignment ($D_2(f, g, \gamma) = D_2(g, f, \gamma^{-1})$).

    \item \textbf{Isometry ($\mathbb{L}^2$-Preserving):}  
    This formulation acts as a variant of the SRVF approach \citep{srivastava2011} applied directly to the signal intensities in the original space. By treating the signals as square-integrable half-densities, we define:
    \begin{equation}
        D_3(f, g, \gamma) = \int_{0}^{1} (f(t) - g(\gamma(t)) \sqrt{\gamma'(t)})^2 dt 
    \end{equation}
    Under this isometric assumption, the mapping preserves the total $\mathbb{L}^2$ energy of the signal under deformation without requiring numerical differentiation of the observed data. While this approach offers geometric elegance and reparameterization invariance, it introduces a specific asymptotic bias in settings where the true generative model follows $f = g \circ \gamma$. In such cases, the $\sqrt{\gamma'}$ term forces a compensatory amplitude scaling that does not exist in the physical observation.

    \item \textbf{Jacobian-Weighted $\mathbb{L}^2$:}  
    Building upon the weighted principles \citep{wang1997alignment}, we utilize the square root of the Jacobian as a weighting factor on the residual itself:
    \begin{equation}
        D_4(f, g, \gamma) = \int_{0}^{1} (f(t) - g(\gamma(t)))^2 \sqrt{\gamma'(t)} \, dt 
        \label{eq:mismatch4}
    \end{equation}
    This formulation serves as a highly effective geometric middle ground. By weighting the squared residual with the square root of the local deformation rate, $\sqrt{\gamma'(t)}$, it dynamically scales the mismatch penalty based on the degree of temporal stretching or compression. 
\end{enumerate}

\subsection{Properties of the Mismatch Formulations}

To ensure a meaningful and numerically stable alignment, we evaluate the four data mismatch options against four essential criteria: Symmetry and Inverse Consistency, Jacobian Role and Pinching Resistance, Self-Alignment Correctness, and Sensitivity to Scale and Translation.

\begin{enumerate}[label=\Roman*.]
    \item \textbf{Symmetry and Inverse Consistency}: In the context of the data mismatch, \textit{symmetry} denotes the invariance of the distance metric under a swap of signal roles and the inversion of the warping function. \textit{Inverse consistency} refers to the reciprocal relationship between the resulting alignments. Our evaluation focuses on whether the mismatch functional strictly satisfies:
    $$
        D(f, g, \gamma) = D(g, f, \gamma^{-1})
    $$
    If this equality holds, the data mismatch term is invariant to the choice of reference signal, ensuring that the data-driven component of the alignment is dictated by the intrinsic geometric relationship between $f$ and $g$.

    The four mismatch formulations behave as follows:
    \begin{enumerate}[label=\roman*)]
        \item \textbf{Standard $\mathbb{L}^2$}: Fails symmetry. In general, $\int (f - g \circ \gamma)^2 \neq \int (g - f \circ \gamma^{-1})^2 $, leading to strong directional bias depending on which signal is chosen as the target.
        
        \item \textbf{Symmetric $\mathbb{L}^2$}: Inherently symmetric by construction. By defining the cost as the explicit sum of the forward and backward residuals, we mathematically guarantee $D_2(f, g, \gamma) = D_2(g, f, \gamma^{-1})$.
        
        \item \textbf{Isometry ($\mathbb{L}^2$-Preserving)}: Satisfies symmetry via the isometry of the square-root half-density transformation, which fundamentally preserves the $\mathbb{L}^2$ inner product under diffeomorphic mapping.
        
\item \textbf{Jacobian-Weighted $\mathbb{L}^2$}: Achieves perfect symmetry through the balanced Jacobian weight. Using the change of variables $\tau = \gamma(t)$, we note that $dt = (\gamma^{-1})'(\tau) d\tau$ and $\gamma'(t) = 1 / (\gamma^{-1})'(\tau)$. Substituting this into the mismatch term in Equation \eqref{eq:mismatch4} yields:
        $$
            D_4(f, g, \gamma) = \int_{0}^{1} (f(\gamma^{-1}(\tau)) - g(\tau))^2 \sqrt{\frac{1}{(\gamma^{-1})'(\tau)}} (\gamma^{-1})'(\tau) d\tau
        $$
        Simplifying the weights provides the exact inverse formulation:
        $$
            D_4(f, g, \gamma) = \int_{0}^{1} (g(\tau) - f(\gamma^{-1}(\tau)))^2 \sqrt{(\gamma^{-1})'(\tau)} d\tau = D_4(g, f, \gamma^{-1})
        $$
    \end{enumerate}

 \item \textbf{The Jacobian Role and Pinching Resistance}: This property examines how the mismatch term interacts with the geometric regularization $\|\psi\|_{\mathbb{H}}^2$. Specifically, we evaluate the tendency of the unregularized data term to encourage ``pinching'' -- defined here as driving the derivative to extreme singularities ($\gamma' \to 0$ or $\gamma' \to \infty$) to artificially hide mismatched features. As will be demonstrated in Section~\ref{subsec:pinching}, the lack of inherent resistance in most of these formulations establishes the absolute necessity of the Sobolev penalty to maintain a valid diffeomorphism.
    \begin{enumerate}[label=\roman*)]
        \item \textbf{Standard $\mathbb{L}^2$}: Susceptible to pinching. Without a Jacobian weight to account for the change in coordinate measure, the optimizer easily bypasses structural mismatches by squeezing ($\gamma' \to 0$) or excessively stretching ($\gamma' \to \infty$) regions of the template. This geometric instability makes the external Sobolev penalty strictly necessary to prevent singular mappings.
        
        \item \textbf{Symmetric $\mathbb{L}^2$}: Susceptible to dual singularities. Contrary to intuition, the $(1 + \gamma')$ weight does not prevent pinching but is actively exploited by the optimizer. When confronted with amplitude mismatches, the unregularized functional minimizes the forward measure by forcing $\gamma' \to 0$ (creating a plateau), which concurrently forces a singular spike ($\gamma^{-1\prime} \to \infty$) in the inverse domain. The penalty is therefore mandatory to prevent these failures.
        
        \item \textbf{Isometry ($\mathbb{L}^2$-Preserving)}: Intrinsically resists pinching but distorts amplitude. Because the mapping applies the $\sqrt{\gamma'}$ Jacobian \textit{inside} the squared residual, extreme squeezing or stretching causes an internal energy explosion, naturally avoiding singularities. However, this structurally intertwines phase and amplitude, acting as a local amplitude scaling factor on $g$ that alters the physical peaks of the data. We enforce the Sobolev penalty here primarily to maintain a unified regularity framework.
        
        \item \textbf{Jacobian-Weighted $\mathbb{L}^2$}: Susceptible to pinching. Because the $\sqrt{\gamma'}$ weight sits \textit{outside} the squared residual, the unregularized optimizer can drive the local mismatch cost strictly to zero simply by setting $\gamma' = 0$ wherever the signals differ. This leads to severe flat regions in the warping function, making the Sobolev penalty absolutely critical to enforce strict monotonicity.
    \end{enumerate}

\item \textbf{Self-Alignment Correctness}: We evaluate the behavior of the data mismatch term when a function is aligned with itself ($f=g$). To be structurally sound, the chosen distance operator $D(f, f, \gamma)$ must uniquely vanish at exactly $\gamma = \gamma_{id}$, ensuring the metric inherently identifies the identity map and strictly penalizes any spurious deformations.

    To formalize this, we restrict our analysis to signals that do not contain flat, constant regions, which would trivially allow infinite non-identity alignments to yield zero mismatch.

\begin{definition}[Reasonable Signal]
    A signal $f \in C([0,1])$ is considered \textit{reasonable} if it is non-constant and can be partitioned into a finite number of sub-intervals upon which it is strictly monotonic.
\end{definition}

    \begin{proposition}[Uniqueness of Self-Alignment]
    Let $f$ be a reasonable signal on $[0, 1]$. If there exists a valid diffeomorphism $\gamma \in \Gamma$ such that $f(t) = f(\gamma(t))$ for all $t \in [0, 1]$, then $\gamma = \gamma_{id}$.
    \label{prop:uni}
    \end{proposition}
     The proof of Proposition \ref{prop:uni} is given in \ref{app:uni_proof}. Given this proposition, the four mismatch formulations on a reasonable signal $f$ behave as follows:
    \begin{enumerate}[label=\roman*)]
        \item \textbf{Standard and Symmetric $\mathbb{L}^2$}: Satisfy correctness. Both $D_1$ and $D_2$ vanish if and only if $f(t) = f(\gamma(t))$, which uniquely implies $\gamma = \gamma_{id}$ for reasonable signals.
        
        \item \textbf{Isometry ($\mathbb{L}^2$-Preserving)}: Satisfies correctness. The condition for $D_3$ to vanish, $f(t) = f(\gamma(t))\sqrt{\gamma'(t)}$, is uniquely satisfied by $\gamma = \gamma_{id}$ for reasonable signals, a foundational property of the SRVF framework \citep{srivastava2011}.
        
        \item \textbf{Jacobian-Weighted $\mathbb{L}^2$}: Satisfies correctness. Because our CLR framework structurally guarantees $\gamma'(t) > 0$ everywhere, the weighted residual in $D_4$ vanishes if and only if the base residual vanishes ($f(t) = f(\gamma(t))$), uniquely implying $\gamma = \gamma_{id}$.
    \end{enumerate}
  
  \item \textbf{Sensitivity to Scale and Translation}: While purely geometric methods often strive for scale and translation invariance, alignment between physical functions frequently requires preserving absolute magnitude and location. We explicitly note that our distance formulations are sensitive to independent scaling and translation (e.g., replacing $g$ with $cg + d$). Because these mismatch terms entangle the structural phase with vertical scale and shifts, they do not perform abstract, scale-independent shape alignment (with the exception of Method 3 (Isometry)'s inherent scale decoupling). Instead, the resulting diffeomorphisms strictly preserve the physical scale and positional integrity of the functions, which is critical for applications where these attributes carry vital domain information.
    
    
\end{enumerate}

\subsection{An Illustrative Example of the Pinching Effect and Singularity Duality}
\label{subsec:pinching}

As discussed in Section 3.2, relying solely on the data mismatch term ($\lambda = 0$) leaves the optimization vulnerable to a geometric instability known as the ``pinching'' effect. Without a Sobolev penalty to enforce smoothness, the functional minimizes mismatch costs by introducing singularities into the warping derivative. While this phenomenon is well-known for the Standard $\mathbb{L}^2$ (Method 1) \citep{ramsay2005functional, srivastava2011}, we provide a concrete toy example here to illustrate exactly how the proposed weighted functionals of the Symmetric $\mathbb{L}^2$ (Method 2) and Jacobian-Weighted $\mathbb{L}^2$ (Method 4) are equally susceptible to these degenerate solutions.

To this end, we define a target signal $f(t)$ and a source $g(t)$ as perfectly phase-aligned triangle waves with a $2:1$ amplitude ratio:
$$
f(t) = \begin{cases} 4t & 0 \le t \le 0.5 \\ 4(1-t) & 0.5 < t \le 1 \end{cases}, \quad 
g(t) = \begin{cases} 2t & 0 \le t \le 0.5 \\ 2(1-t) & 0.5 < t \le 1 \end{cases}
$$
The identity warping $\gamma_{id}(t) = t$, while providing the true, intuitive geometric matching between the peaks of $f$ and $g$, cannot minimize the objective function in the absence of a penalty. Because $g$ cannot reach the vertical range of $f$, the registration using the identity warp yields an unavoidable vertical residual. We examine the behavior of the optimal warping function $\gamma$ for Method 2 with $\lambda = 0$, as illustrated in Figure~\ref{fig:pinching_plot}.

In this unregularized case, the optimization actively abandons the correct structural alignment to chase the amplitude difference. The resulting optimal warping contains flat regions ($\gamma'=0$), and thus ceases to be a well-defined diffeomorphism; strictly speaking, its inverse does not exist. However, in a limiting sense, these singular forms illustrate the infimum reachable by the unregularized mismatch term alone. They are presented here for intuition, and the rigorous mathematical details of this limiting behavior are omitted for brevity.

\begin{figure}[htbp]
    \centering
    \includegraphics[width=\textwidth]{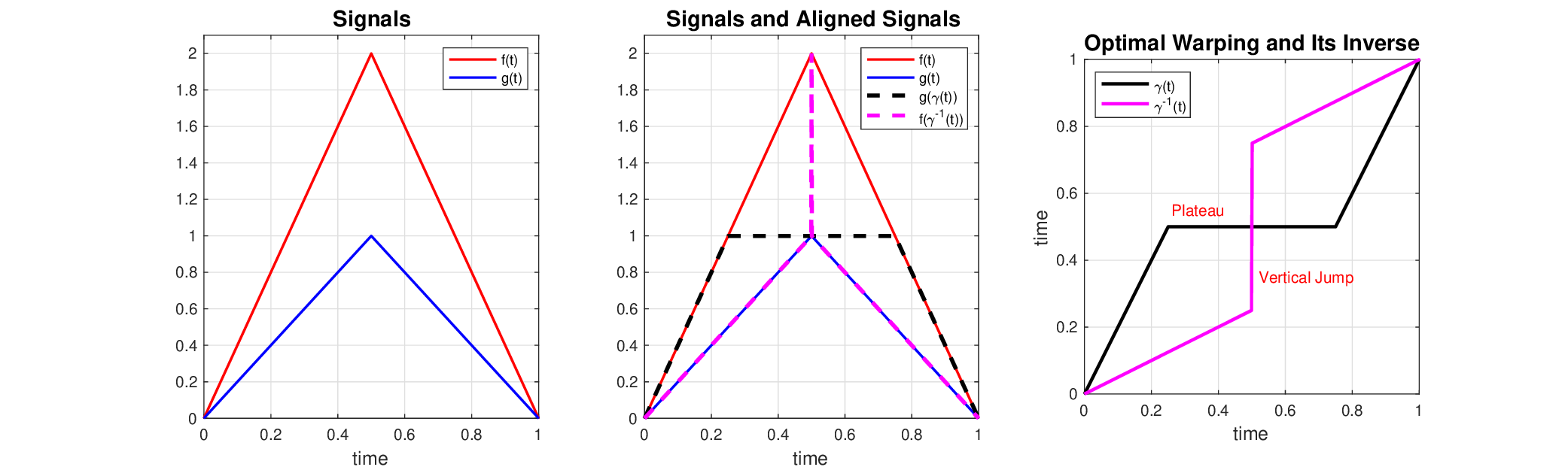} 
    \caption{Geometric singularities in unregularized Symmetric $\mathbb L^2$ registration. (Left) Original signals with $2:1$ amplitude ratio. (Center) Aligned signals: $g(\gamma(t))$ (black) plateaus at its maximum height, while the inverse alignment $f(\gamma^{-1}(t))$ (magenta) pinches the target's peak into a singular spike. (Right) The optimal warping $\gamma$ and its inverse $\gamma^{-1}$ exhibit a horizontal plateau and a reciprocal vertical jump, respectively.}
    \label{fig:pinching_plot}
\end{figure}

The results reveal a critical duality in registration failure. To minimize the weighted error $\int (f(t) - g(\gamma(t))^2(1 + \gamma'(t)) dt = \int_{0}^{1} (f(t) - g(\gamma(t)))^2 + (g(t) - f(\gamma^{-1}(t)))^2 dt$, the optimizer exploits the derivative weight in two distinct ways:
\begin{enumerate}
    \item Forward Domain (The Plateau): To minimize the residual at the peak, the optimizer sets $\gamma'=0$ over the interval $t \in [1/4, 3/4]$ where $f > g$. This effectively ``stops the clock'' at the peak of $g$. This results in the black dashed signal in Figure~\ref{fig:pinching_plot} (Center).
    
    \item Inverse Domain (The Spike): Due to the symmetric nature of the framework, the inverse warp $\gamma^{-1}$ must compensate for the forward plateau with a \textit{vertical jump} (infinite derivative). Consequently, the inverse-aligned signal $f(\gamma^{-1}(t))$ compresses the entire energy of the target's peak into a single point, creating the Dirac-delta-like spike shown in magenta.
\end{enumerate}

The transition between the linear segments and the plateau in Figure~\ref{fig:pinching_plot} (Right) creates sharp corners where the first derivative $\gamma'$ has a jump discontinuity. In Method 4, this instability is even more pronounced, as the $\sqrt{\gamma'}$ term allows the infimum of the cost to reach zero by setting $\gamma'=0$ whenever a mismatch occurs.
These results demonstrate that for Methods 1, 2, and 4, the optimal limiting solution is not a true diffeomorphism but a piecewise-differentiable, singular mapping; indeed, these pinching warpings represent the infimum of the unregularized objective functionals.

In contrast, Method 3 (a SRVF variant) is intrinsically regularized. Because the Jacobian $\sqrt{\gamma'}$ is coupled inside the square of the integrand, the formulation acts as an isometry that preserves total signal energy. Consequently, any tendency toward $\gamma' \to 0$ or $\gamma' \to \infty$ incurs a strict energetic penalty. Thus, for Method 3, the pinching effect is naturally avoided even without an external penalty. However, to maintain a consistent regularity framework, we enforce the Sobolev penalty across all four methods. This ensures that the resulting warps are smooth diffeomorphisms regardless of the chosen mismatch term.

\subsection{Summary of the Proposed Mismatch Formulations}
\label{subsec:framework_summary}

We provide a comprehensive summary of the four mismatch formulations analyzed in this work. Table~\ref{tab:framework_comparison} compares the intrinsic mathematical properties of these unregularized data fidelity terms, revealing their inherent structural differences. The table highlights the critical trade-offs between physical intuition, bidirectional symmetry, and intrinsic robustness to numerical artifacts like the ``pinching'' effect. Most importantly, this comparison illustrates exactly why Methods 1, 2, and 4 rely heavily on an external Sobolev penalty to maintain geometric stability, whereas Method 3 intrinsically penalizes singular behavior.

\begin{table}[ht]
\centering
\caption{Comprehensive Comparison of the Unregularized Mismatch Formulations}
\label{tab:framework_comparison}
\scriptsize
\renewcommand{\arraystretch}{1.8} 
\begin{tabular}{|p{2.8cm}|c|c|c|c|}
\hline
\textbf{Method} & \textbf{1. Standard} & \textbf{2. Symmetric} & \textbf{3. Isometry} & \textbf{4. Jacobian-W} \\ \hline
\textbf{Mismatch Term} & $\int_{0}^{1} (f - g \circ \gamma)^2$ & $\int_{0}^{1} (f - g \circ \gamma)^2 \left(\frac{1 + \gamma'}{2}\right)$ & $\int_{0}^{1} (f - (g \circ \gamma)\sqrt{\gamma'})^2$ & $\int_{0}^{1} (f - g \circ \gamma)^2 \sqrt{\gamma'}$ \\ \hline
\textbf{Interpretability} & Intuitive & Intuitive & Physical & Mathematical \\ \hline
\textbf{Symmetry} & No & Yes & Yes & Yes \\ \hline
\textbf{Inv. Consistency} & No & Yes & Yes & Yes \\ \hline
\textbf{Intrinsic Pinching Resist.} & Low & Low & High & Very Low \\ \hline
\textbf{Regularization} & Necessary & Necessary & Intrinsic & Necessary \\ \hline
\textbf{Self-Alignment} & Correct & Correct & Correct & Correct \\ \hline
\end{tabular}
\end{table}

Moving down the properties in Table~\ref{tab:framework_comparison}, the \textit{Interpretability} row highlights fundamentally different conceptual approaches. Methods 1 and 2 are classified as \textit{Intuitive} because they rely on the direct vertical distance between signal values. Method 3, described as \textit{Physical}, treats the signal as a half-density subject to preserve the $\mathbb{L}^2$ norm by incorporating $\sqrt{\gamma'}$ inside the square. Method 4 is characterized as \textit{Mathematical}; rather than modifying the signal amplitude itself, it uses the $\sqrt{\gamma'}$ weight to adjust the underlying integration measure required for reparameterization invariance.

Regarding \textit{Symmetry} and \textit{Inverse Consistency}, bidirectional agreement is critical for unbiased registration. Method 1 strictly computes a unidirectional mapping, failing to satisfy these properties. In contrast, Methods 2, 3, and 4 explicitly enforce a symmetric structural formulation. This ensures that the registration is invariant to the choice of the template, guaranteeing that the optimal mapping from $f$ to $g$ is exactly the inverse of the mapping from $g$ to $f$.

A major differentiator among these models is their \textit{Intrinsic Pinching Resistance}, which directly dictates their reliance on external Sobolev \textit{Regularization}. Methods 1 and 2 exhibit low intrinsic resistance; without a dominant penalty, they are susceptible to collapsing mismatched regions into a near-zero measure. Method 4 is mathematically even more vulnerable (Very Low resistance). Consequently, Methods 1, 2, and 4 rely entirely on the Sobolev penalty to act as a strict geometric guard (making external regularization \textit{Necessary}). Method 3, however, possesses high intrinsic resistance. Any attempt to pinch the domain leaves a massive residual energy, inherently preventing singular plateaus and making its geometric stability \textit{Intrinsic}.

Finally, despite these differences in pinching vulnerability and geometric interpretation, all four formulations successfully guarantee mathematically correct \textit{Self-Alignment}. As long as the signals are reasonable (lacking flat, constant regions), each mismatch term uniquely vanishes exactly at the identity map ($\gamma = \gamma_{id}$), ensuring no spurious deformations are introduced when a function is aligned with itself.

\subsection{Existence of the Optimal Solution}

Having established the distinct mathematical properties of the unregularized mismatch formulations, we now return to the complete registration framework. To conclude our theoretical development, we must guarantee that combining these data fidelity terms with the Sobolev regularity penalty yields a mathematically well-posed optimization problem. We establish the existence of a global minimizer $\psi^* \in \mathbb{H}$ for the full functional registration objective using the Direct Method in the Calculus of Variations. Throughout this section, we assume the raw signals are continuous, $f, g \in C^0(I)$, to ensure well-posed point-wise evaluations.

It is critical to note that the Sobolev penalty $\lambda \|\psi\|_{\mathbb{H}}^2$ is not merely a smoothing term, but a fundamental requirement for the well-posedness of the problem. As demonstrated in Section \ref{subsec:pinching}, removing this regularization ($\lambda = 0$) causes the objective functionals to lose their coercivity, allowing the optimization to collapse into singular, piecewise-differentiable mappings. The penalty prevents this geometric instability by ensuring the optimization remains within a bounded subset of $\mathbb{H}$ that embeds compactly into $C^1(I)$. This theoretical framework guarantees that a physically meaningful, strictly diffeomorphic warping function exists, providing a rigorously defined global minimum for our numerical search.

\begin{enumerate}[label=\Roman*.]
    \item \textbf{Standard $\mathbb{L}^2$ Mismatch (Baseline):} We utilize the Sobolev space $\mathbb{H} \subset C^1(I)$, a Hilbert space of zero-mean functions with square-integrable first and second derivatives. As established in Section \ref{sec:space}, the space is equipped with the Sobolev roughness penalty $\|\psi\|_{\mathbb{H}}^2 = \int_0^1 (\psi'(t))^2 dt + \int_0^1 (\psi''(t))^2 dt$.

\begin{theorem}[Existence of Optimal Warping]
Let $f, g \in C^0(I)$. There exists a minimizer $\psi^* \in \mathbb{H}$ for the objective functional:
$$
    \mathcal{O}_1(\psi) = \int_{0}^{1} (f(t) - g(\gamma_\psi(t)))^2 dt + \lambda \|\psi\|_{\mathbb{H}}^2
$$
where $\gamma_\psi = \Phi^{-1}(\psi)$ is the normalized exponential map of $\psi$ (in Equation \eqref{eq:invCLR}).
\end{theorem}

Establishing the existence of this minimizer ensures the functional registration problem is well-posed in the continuous limit, providing a rigorous theoretical target that justifies our subsequent numerical search within a finite-dimensional subspace. The Sobolev penalty, weighted by $\lambda$, acts as a strict geometric guard against the infinite-dimensional ``pinching'' phenomenon ($\gamma' \to 0$ or $\gamma' \to \infty$), ensuring that the resulting warping function remains a strictly smooth diffeomorphism. A detailed proof utilizing the reflexivity of the Sobolev space and its compact embedding into $C^1(I)$ is provided in \ref{app:existence_proof1}.

\item \textbf{Symmetric $\mathbb{L}^2$ Mismatch:}
Beyond the standard coordinate shift, the Symmetric $\mathbb{L}^2$ framework incorporates the change-of-variables Jacobian directly into the mismatch functional to ensure bidirectional alignment.

\begin{theorem}[Existence of Optimal Symmetric Warping]
Let $f, g \in C^0(I)$. There exists a minimizer $\psi^* \in \mathbb{H}$ for the symmetric objective:
$$
    \mathcal{O}_2(\psi) = \int_{0}^{1} (f(t) - g(\gamma_\psi(t)))^2 \left (\frac{1 + \gamma_\psi'(t)}{2} \right ) dt + \lambda \|\psi\|_{\mathbb{H}}^2
$$
\end{theorem}

The existence of a solution for the symmetric case is mathematically non-trivial because the functional depends explicitly on the derivative $\gamma'$ within the data-matching integration weight. As demonstrated in Section 3.3, the inclusion of $\gamma'$ in the mismatch formulation cannot, on its own, prevent degenerate warping. Therefore, the Sobolev regularity framework serves a critical dual purpose here. First, the penalty steps in to strictly bound the warping derivative away from zero and infinity, actively preventing the singular "pinching" phenomenon that the data term is blind to. Second, the compact embedding of $\mathbb{H}$ into $C^1(I)$ ensures that this derivative field converges uniformly. This guarantees that the symmetric integration weight, $\frac{1 + \gamma_\psi'(t)}{2}$, remains stable during the variational limiting process, ultimately securing a global minimizer as detailed in \ref{app:existence_proof2}.

\item \textbf{Isometry ($\mathbb{L}^2$-Preserving) Mismatch:}
The isometry-based objective functional treats the signals as half-densities, utilizing the square root of the Jacobian to preserve the $\mathbb{L}^2$ norm:
$$
    \mathcal{O}_{3}(\psi) = \int_{0}^{1} (f(t) - g(\gamma_\psi(t))\sqrt{\gamma_\psi'(t)})^2 dt + \lambda \|\psi\|_{\mathbb{H}}^2
$$

\begin{theorem}[Existence of Isometry-Based Optimal Warping]
Let $f, g \in C^0(I)$. There exists a minimizer $\psi^* \in \mathbb{H}$ for the functional $\mathcal{O}_{3}(\psi)$.
\end{theorem}

The existence proof for the Isometry case relies on the stability of the square-root Jacobian operator under the $C^1$ topology. Because the CLR representation is constrained within the Sobolev space $\mathbb{H}$, the resulting derivative $\gamma'$ is strictly bounded away from zero. This guarantees the uniform convergence of the $\sqrt{\gamma'}$ term during the variational limit. The detailed variational analysis is provided in \ref{app:existence_proof3}.

\item \textbf{Jacobian-Weighted $\mathbb{L}^2$ Mismatch:}
The objective functional for the Jacobian-weighted framework incorporates the square root of the warping derivative as a local scaling factor to maintain symmetry without assuming signal isometry:
$$
    \mathcal{O}_{4}(\psi) = \int_{0}^{1} (f(t) - g(\gamma_\psi(t)))^2 \sqrt{\gamma_\psi'(t)} \, dt + \lambda \|\psi\|_{\mathbb{H}}^2
$$

\begin{theorem}[Existence of Jacobian-Weighted Optimal Warping]
Let $f, g \in C^0(I)$. There exists a minimizer $\psi^* \in \mathbb{H}$ for the functional $\mathcal{O}_4(\psi)$.
\end{theorem}

The existence of a minimizer for the Jacobian-weighted framework follows from the uniform convergence of the weighted residual. Unlike the Standard $\mathbb{L}^2$ approach, Method 4 weights the signal mismatch by the local stretching factor. As demonstrated in Section \ref{subsec:pinching}, this specific weighting makes the unregularized functional highly susceptible to extreme pinching. However, the regularity of the CLR representation within the Sobolev space $\mathbb{H}$ ensures that this integration weight remains strictly positive and bounded. This strictly prevents the artificial collapse of the objective cost and guarantees a well-defined global minimizer. The step-by-step variational proof is detailed in \ref{app:existence_proof4}.
\end{enumerate}

\section{Computational Methodology and Comparative Results}

In this section, we transition from the continuous geometric framework established in Section \ref{sec:sobolev} to its numerical realization. We evaluate the computational performance of the four registration formulations by projecting the continuous variational problem onto finite-dimensional subspaces of the Sobolev space $\mathbb{H}$.

\subsection{Numerical Implementation via Finite Basis Expansion}

To maintain consistency with our theoretical guarantees while achieving computational feasibility, we perform the optimization by projecting the continuous variational problem from the infinite-dimensional Sobolev space $\mathbb{H}$ onto a finite-dimensional subspace $\mathbb{H}_{d} \subset \mathbb{H}$. This projection allows us to conduct highly efficient, unconstrained optimization over a set of basis coefficients in a standard Euclidean parameter space ($\mathbb{R}^d$). Importantly, because the basis functions themselves are continuous, this approach strictly preserves the diffeomorphism properties of the warping manifold $\Gamma$ without requiring artificial discretization of the time domain.

\subsubsection{Parametrization of the CLR Field}

We begin by expressing the unconstrained CLR representation $\psi(t) \in \mathbb{H}$ as a linear combination of a chosen set of continuous, zero-mean basis functions $\{\phi_j(t)\}_{j=1}^{d}$ that span the subspace $\mathbb{H}_{d}$. To emphasize the dependence of the continuous field on the underlying coefficients, we denote this finite-dimensional approximation as $\psi_{\mathbf{c}}(t)$:
$$
    \psi_{\mathbf{c}}(t) = \sum_{j=1}^{d} c_j \phi_j(t) = \mathbf{c}^\top \boldsymbol{\phi}(t)
$$
where $\mathbf{c} = (c_1, \dots, c_d)^\top \in \mathbb{R}^{d}$ is the coefficient vector to be optimized, and $\boldsymbol{\phi}(t) = (\phi_1(t), \dots, \phi_d(t))^\top$ represents the vector-valued basis function.

While this finite-basis framework is highly generic and can be successfully implemented using various basis systems, such as Bernstein polynomials or truncated Fourier series, we primarily focus our methodological exposition on cubic B-spline bases \citep{deBoor2001}. Unlike global bases that span the entire domain and yield dense penalty matrices, B-splines possess strictly local support. This property inherently generates highly sparse, banded penalty matrices at higher dimensions, drastically reducing the computational complexity of the optimization.

For our later theoretical guarantees, we require the sequence of subspaces to be \textit{dense} in $\mathbb{H}$. That is, as the subspace dimension $d \to \infty$, any true CLR representation $\psi \in \mathbb{H}$ can be expressed as the limit of an approximating sequence $\tilde{\psi}_d \in \mathbb{H}_{d}$ such that $\|\tilde{\psi}_d - \psi\|_{\mathbb{H}} \to 0$.

The mapping from this finite-dimensional coefficient space to the warping manifold $\Gamma$ is performed via the inverse CLR transform $\Phi^{-1}$. The resulting warping function and its derivative are exact continuous functionals of the basis expansion $\psi_{\mathbf{c}}$:
$$
     \gamma_{\mathbf{c}}(t) = \Phi^{-1}(\psi_{\mathbf{c}})(t) = \frac{\int_0^t \exp(\psi_{\mathbf{c}}(s)) ds}{\int_0^1 \exp(\psi_{\mathbf{c}}(s)) ds}, \quad \gamma'_{\mathbf{c}}(t) = \frac{\exp(\psi_{\mathbf{c}}(t))}{\int_0^1 \exp(\psi_{\mathbf{c}}(s)) ds} 
$$
This formulation ensures that the boundary conditions $\gamma_{\mathbf{c}}(0)=0$ and $\gamma_{\mathbf{c}}(1)=1$ are satisfied by construction. Furthermore, because the basis functions are explicitly centered (ensuring $\psi_{\mathbf{c}}$ strictly maintains the zero-mean property of $\mathbb{H}$), the mapping between the finite-dimensional coefficient space $\mathbb{R}^d$ and the approximated warping manifold remains strictly bijective.

\subsubsection{Finite-Dimensional Objective and Sobolev Regularization}
\label{subsec:finite_obj}

The numerical registration is conducted by evaluating the infinite-dimensional functionals purely within the finite-dimensional subspace $\mathbb{H}_d$. For any of the formulations $i \in \{1,2,3,4\}$ summarized in Section \ref{subsec:framework_summary}, the objective function $J(\mathbf{c})$ evaluates the mismatch functional, defined as $\mathcal{D}_i(\psi) =  D_i(f, g, \Phi^{-1}(\psi))$, and the Sobolev norm at the continuous basis approximation $\psi_{\mathbf{c}}$:
\begin{equation}
    J(\mathbf{c}) = \mathcal{D}_i(\psi_{\mathbf{c}}) + \lambda \|\psi_{\mathbf{c}}\|_{\mathbb{H}}^2
    \label{eq:finite_obj}
\end{equation}

A major computational advantage of this representation is that the calculation of the Hilbert norm $\|\psi\|_{\mathbb{H}}^2$ (the roughness penalty) transforms from a functional integral into an exact quadratic form. Incorporating both the first and second derivative penalties defined by our Sobolev space, the objective becomes:
$$
    J(\mathbf{c}) = \mathcal{D}_i(\psi_{\mathbf{c}}) + \lambda \mathbf{c}^\top \mathbf{R} \mathbf{c}
$$
where the total stiffness matrix $\mathbf{R} \in \mathbb{R}^{d \times d}$ has constant entries defined by the chosen basis:
$$
R_{ij} = \int_0^1 \phi_i'(t) \phi_j'(t) dt + \int_0^1 \phi_i''(t) \phi_j''(t) dt, \quad i,j \in \{1, \dots, d\}
$$
This equality ensures that within the basis subspace, the regularization is evaluated without numerical integration error during the optimization loop, provided the entries of $\mathbf{R}$ are pre-computed exactly (e.g., via closed-form integrals or high-precision quadrature of the basis derivatives). 

From a computational perspective, this exact quadratic formulation provides crucial advantages for robust functional alignment. By utilizing the stiffness matrix $\mathbf{R}$ to penalize the continuous derivatives of the approximated field $\psi_{\mathbf{c}}$, the optimizer implicitly bounds the curvature of the resulting warping function. As discussed in Section \ref{subsec:pinching}, this Sobolev regularization is the primary mathematical defense against the geometric instability known as the ``pinching'' effect. Any emerging singularity in the warping function $\gamma_{\mathbf{c}}$ would require the derivatives of $\psi_{\mathbf{c}}$ to approach infinity, which is strictly penalized by the quadratic form $\lambda \mathbf{c}^\top \mathbf{R} \mathbf{c}$. Furthermore, evaluating this penalty as a simple, deterministic matrix-vector product avoids the severe computational costs associated with repeated numerical integration during each optimization step, ensuring a highly scalable and robust registration framework.

\subsection{Noise-Free Asymptotic Consistency}
\label{subsec:consis}

Having defined the objective functionals and their regularized finite-basis representations, we now formalize their theoretical validity. Before presenting the asymptotic properties of our framework, we clarify the scope of our analysis. Our consistency theory focuses exclusively on the M-estimator -- defined as the global minimizer of the finite-dimensional regularized objective $J(\mathbf{c})$ given in Equation~\eqref{eq:finite_obj} -- which serves as a restricted subspace approximation of the infinite-dimensional Sobolev space $\mathbb{H}$. Throughout this analysis, we assume continuous observation of the functional signals, meaning the matching functionals (e.g., $\mathcal{D}_i$) are evaluated continuously without discretization error. Consequently, our theoretical guarantees are intrinsic properties of the M-estimator itself, independent of the specific numerical integration schemes or optimization algorithms utilized in the subsequent computational implementation.

We consider a noise-free setting where there exists a true, infinite-dimensional ground-truth warping function $\gamma_0$ (with corresponding CLR representation $\psi_0 \in \mathbb{H}$) that perfectly aligns the two signals, such that $f(t) = g(\gamma_0(t))$. To rigorously bridge our finite-dimensional approximation with the infinite-dimensional theoretical ideal, our analysis proceeds in two distinct steps. First, we guarantee that a global minimizer actually exists for Equation~\eqref{eq:finite_obj} within our restricted $d$-dimensional subspace, ensuring our objective is well-posed. Second, we demonstrate that as the subspace dimension grows ($d \to \infty$) and the Sobolev smoothing parameter decays ($\lambda_d \to 0$) at appropriate rates, the finite-dimensional estimator converges strictly to the true underlying warping function $\gamma_0$ in the Sobolev norm.

\subsubsection{Existence of the Finite-Dimensional Minimizer}

Before establishing asymptotic convergence, we must guarantee that the finite-dimensional objective actually attains a minimum. Recall the regularized objective function introduced in Equation~\eqref{eq:finite_obj}. To formally analyze its behavior as the subspace expands, we explicitly index the objective and the smoothing parameter by the dimension $d$, expressing it as:
\begin{equation}
    J_d(\mathbf{c}) = \mathcal{D}_i(\psi_{\mathbf{c}}) + \lambda_d \|\psi_{\mathbf{c}}\|_{\mathbb{H}}^2
    \label{eq:finite_obj_d}
\end{equation}
Here, $\mathbf{c} \in \mathbb{R}^{d}$ represents the coefficients of the $d$-dimensional CLR representation $\psi_{\mathbf{c}}$, and $i \in \{1,2,3,4\}$ designates the chosen registration mismatch. By linking the penalty $\lambda_d$ to the dimension, we can properly analyze the regularization decay as the basis resolution increases. The existence of a global minimizer for this sequence of formulations is established in Lemma~\ref{lemma:exist} (see detailed proof in \ref{app:consistency_proof0}).

\begin{lemma}[Existence of the Minimum] 
For any fixed $d < \infty$, $\lambda_d > 0$, and $i \in \{1, \dots, 4\}$, there exists at least one global minimizer $\hat{\mathbf{c}}_d \in \mathbb{R}^{d}$ that minimizes $J_d(\mathbf{c})$ in Equation \eqref{eq:finite_obj_d}.
\label{lemma:exist}
\end{lemma}

With the existence of the finite-dimensional minimizer mathematically guaranteed, we now establish that this estimator asymptotically converges to the true infinite-dimensional warping function as the subspace expands.

\subsubsection{Asymptotic Consistency in Method 1: Standard $\mathbb{L}^2$}

Let $\mathbb{H}_{d} \subset \mathbb{H}$ be the finite-dimensional subspace of dimension $d$ (spanned by $d$ basis functions). The standard $\mathbb{L}^2$ data mismatch functional $\mathcal{D}_1(\psi)$ evaluates the integrated squared error between the target signal and the warped source signal. For any given $\psi \in \mathbb{H}_d$ and $\gamma_\psi = \Phi^{-1}(\psi)$, this takes the explicit form:
$$
    \mathcal{D}_1(\psi) = \int_{0}^{1} (f(t) - g(\gamma_\psi(t)))^2 dt
$$
Using this mismatch term, the finite-dimensional estimator is then defined as the minimizer of the regularized objective over the subspace $\mathbb{H}_{d}$:
\begin{equation}
    \hat{\psi}_d = \arg \min_{\psi \in \mathbb{H}_{d}} \left\{ \mathcal{D}_1(\psi) + \lambda_d \|\psi\|_{\mathbb{H}}^2 \right\}
    \label{eq:con1}
\end{equation}

\begin{proposition}[Consistency of Method 1]
Assume $f$ is a reasonable function and $g = f \circ \gamma_0^{-1}$, where $\gamma_0 = \Phi^{-1}(\psi_0)$ is the true warping function. Let $\tilde{\psi}_d \in \mathbb{H}_{d}$ be an approximating sequence converging to $\psi_0$. As $d \to \infty$, if $\lambda_d \to 0$ such that $\mathcal{D}_1(\tilde{\psi}_d) = o(\lambda_d)$, then the finite-dimensional estimator $\hat{\psi}_d$ in Equation \eqref{eq:con1} converges to $\psi_0$ in the $\mathbb{H}$ norm.
\label{prop:consis1}
\end{proposition}

The consistency relies on balancing the subspace approximation error with the penalty decay. The condition $\mathcal{D}_1(\tilde{\psi}_d) = o(\lambda_d)$ ensures that $\lambda_d$ vanishes slower than the approximation error, preventing overfitting and guaranteeing the sequence remains bounded in the Sobolev norm. Combined with the identifiability condition (i.e., $f$ is a \textit{reasonable} function), this allows us to rigorously establish the finite-basis estimator as an asymptotically exact recovery of the true diffeomorphism (see detailed proof in \ref{app:consistency_proof1}).

\subsubsection{Asymptotic Consistency in Method 2: Symmetric $\mathbb{L}^2$}

In the symmetric formulation, the data mismatch functional  $\mathcal{D}_2(\psi)$ accounts for both the forward and backward registration residuals:
$$
    \mathcal{D}_2(\psi) = \int_{0}^{1} (f(t) - g(\gamma_\psi(t)))^2 \left(\frac{1 + \gamma'_\psi(t)}{2}\right) dt
$$
Using this balanced mismatch, the finite-dimensional estimator for Method 2 is defined as the minimizer of the regularized objective over the subspace $\mathbb{H}_{d}$:
\begin{equation}
    \hat{\psi}_d = \arg \min_{\psi \in \mathbb{H}_{d}} \left\{ \mathcal{D}_2(\psi) + \lambda_d \|\psi\|_{\mathbb{H}}^2 \right\}
    \label{eq:con2}
\end{equation}

\begin{proposition}[Consistency of Method 2]
Assume $f$ and $g$ are reasonable functions such that $g = f \circ \gamma_0^{-1}$, where $\gamma_0 = \Phi^{-1}(\psi_0)$ is the true warping function. Let $\tilde{\psi}_d \in \mathbb{H}_{d}$ be an approximating sequence converging to $\psi_0$. As $d \to \infty$, if $\lambda_d \to 0$ such that the subspace approximation error satisfies $\mathcal{D}_2(\tilde{\psi}_d) = o(\lambda_d)$, then the symmetric finite-dimensional estimator $\hat{\psi}_d$ in Equation \eqref{eq:con2} converges to $\psi_0$ in the $\mathbb{H}$ norm.
\label{prop:consis2}
\end{proposition}

The consistency of this Symmetric $\mathbb{L}^2$ estimator is established by demonstrating that the symmetric mismatch functional $\mathcal{D}_2$ acts as a dual constraint. While Method 1 relies solely on the forward residual, the symmetric formulation, weighting the residual by the localized stretching factor $(1 + \gamma'_\psi(t))/2$, ensures that both the forward warping and its inverse converge simultaneously to their respective ground truths. This leads to an identical asymptotic limit as Method 1 but provides a more constrained and geometrically balanced optimization path, as detailed in the proof in \ref{app:consistency_proof2}.

\subsubsection{Asymptotic Bias in Method 3: Isometry ($\mathbb{L}^2$-Preserving)}

While the isometry-based formulation in Method 3 possesses desirable geometric properties, it is important to note that the resulting finite-dimensional estimator is \textit{inconsistent} under the standard pure-warping assumption $f = g \circ \gamma_0$. In this setting, the mismatch functional is:
$$
    \mathcal{D}_3(\psi) = \int_{0}^{1} \left(f(t) - g(\gamma_\psi(t))\sqrt{\gamma'_\psi(t)}\right)^2 dt
$$

To demonstrate this lack of consistency, we evaluate the asymptotic mismatch at the true underlying warping function $\gamma_0$. Substituting the true generative model $f(t) = g(\gamma_0(t))$ into the functional yields:
$$
    \mathcal{D}_3(\psi_0) = \int_{0}^{1} \left(g(\gamma_0(t)) - g(\gamma_0(t))\sqrt{\gamma'_0(t)}\right)^2 dt 
    = \int_{0}^{1} (g(\gamma_0(t)))^2 \left(1 - \sqrt{\gamma'_0(t)}\right)^2 dt
$$

For any non-identity true warping where $\gamma'_0(t) \neq 1$ on a set of strictly positive measure (assuming the signal $g$ is not trivially zero), the integrand is strictly positive, meaning $\mathcal{D}_3(\psi_0) > 0$. Because the asymptotic mismatch does not vanish at the ground truth $\psi_0$, the true parameter cannot be the global minimizer of the limit objective. Consequently, as the subspace dimension grows ($d \to \infty$) and the penalty decays ($\lambda_d \to 0$), the sequence of finite-dimensional estimators $\{\hat{\psi}_d\}$ cannot converge to the true warping $\psi_0$. 
Whether the sequence of estimators diverges or converges to some alternative, biased minimizer, the fundamental conclusion remains the same: the $\mathbb{L}^2$-preserving objective is formally inconsistent for pure registration tasks. 

\subsubsection{Asymptotic Consistency in Method 4: Jacobian-Weighted $\mathbb{L}^2$}

In this Jacobian-weighted formulation, the mismatch functional $\mathcal{D}_4(\psi)$ serves as a geometric modification to the standard $\mathbb{L}^2$ distance by incorporating the square root of the warping velocity as a local volume measure on the residual:
$$
    \mathcal{D}_4(\psi) = \int_{0}^{1} (f(t) - g(\gamma_\psi(t)))^2 \sqrt{\gamma'_\psi(t)} dt
$$
Using this modified mismatch, the finite-dimensional estimator for Method 4 is defined as the minimizer within the restricted subspace $\mathbb{H}_{d}$:
\begin{equation}
    \hat{\psi}_d = \arg \min_{\psi \in \mathbb{H}_{d}} \left\{ \mathcal{D}_4(\psi) + \lambda_d \|\psi\|_{\mathbb{H}}^2 \right\}
    \label{eq:con4}
\end{equation}

Unlike the isometry-preserving method, this Jacobian framework evaluates exactly to zero at the true warping $\gamma_0$, preserving the identifiability of the original signal profile. This property is formalized in the following proposition, and a detailed proof is provided in \ref{app:consistency_proof4}.

\begin{proposition}[Consistency of Method 4]
Assume $f$ and $g$ are reasonable functions such that the ground-truth warping satisfies $f = g \circ \gamma_0$ for some true parameter $\gamma_0 = \Phi^{-1}(\psi_0)$. Let $\tilde{\psi}_d \in \mathbb{H}_{d}$ be an approximating sequence converging to $\psi_0$. As $d \to \infty$, if $\lambda_d \to 0$ such that the subspace approximation error satisfies $\mathcal{D}_4(\tilde{\psi}_d) = o(\lambda_d)$, then the Jacobian-weighted finite-dimensional estimator $\hat{\psi}_d$ in Equation \eqref{eq:con4} converges to $\psi_0$ in the $\mathbb{H}$ norm.
\label{prop:consis4}
\end{proposition}

\subsection{Unified Gradient Descent and Computational Algorithm}
\label{subsec:gradient_descent}

Having established the theoretical guarantees and asymptotic behavior of our finite-dimensional estimators across the four registration paradigms, we now detail the explicit computational algorithm used to iteratively solve for the optimal alignment. By restricting our search space to the finite-dimensional subspace $\mathbb{H}_d$, the infinite-dimensional continuous functional optimization is directly projected into a discrete parameter estimation problem.

\subsubsection{Numerical Optimization and Gradient Flow}

To minimize the objective functional $J(\mathbf{c})$, we must bridge variational calculus with finite-dimensional optimization. The sensitivity of the mismatch functional $\mathcal{D}_i(\psi)$ to a local perturbation $h \in \mathbb{H}$ is characterized by the Fréchet derivative $\frac{\delta \mathcal{D}_i}{\delta \psi}$. Formally, for a given framework, this derivative is the unique operator such that:
$$
    \mathcal{D}_i(\psi + h) = \mathcal{D}_i(\psi) + \int_{0}^{1} \frac{\delta \mathcal{D}_i}{\delta \psi}(t) h(t) dt + o(\|h\|_{\mathbb{H}})
$$

By projecting this functional gradient onto our finite-dimensional basis $\{\phi_j\}_{j=1}^{d}$, we transform the variational problem into a standard vector optimization.  The exact continuous gradient of the objective $J$ with respect to the basis coefficients $\mathbf{c} \in \mathbb{R}^{d}$ is given by the integral projection:
$$
    \nabla_{\mathbf{c}} J = \int_0^1 \frac{\delta \mathcal{D}_i}{\delta \psi}(t) \boldsymbol{\phi}(t) dt + 2\lambda \mathbf{R} \mathbf{c}
$$
where the first term represents the data-driven functional force and $2\lambda \mathbf{R} \mathbf{c}$ is the exact linear gradient of the Sobolev regularization.

In computational practice, the integral defining this data-driven gradient must be approximated using numerical quadrature over a finely sampled temporal grid $t \in \{t_1, t_2, \dots, t_N\}$. Evaluated at the continuous basis approximation $\psi_{\mathbf{c}}$, the computational gradient becomes a highly efficient matrix-vector operation:
$$
    \nabla_{\mathbf{c}} J \approx \mathbf{B}^\top \left( \frac{\delta \mathcal{D}_i}{\delta \psi} \Big|_{\psi = \psi_{\mathbf{c}}} \right) + 2\lambda \mathbf{R} \mathbf{c}
$$
where $\mathbf{B} \in \mathbb{R}^{N \times d}$ is the \textbf{Basis Evaluation Matrix} (or design matrix) that evaluates the $d$ continuous basis functions across the integration grid. Explicitly, $\mathbf{B}$ is constructed as:
$$
    \mathbf{B} = 
    \begin{bmatrix} 
    \phi_1(t_1) & \phi_2(t_1) & \dots & \phi_{d}(t_1) \\
    \phi_1(t_2) & \phi_2(t_2) & \dots & \phi_{d}(t_2) \\
    \vdots & \vdots & \ddots & \vdots \\
    \phi_1(t_N) & \phi_2(t_N) & \dots & \phi_{d}(t_N)
    \end{bmatrix}
$$

In this formulation, the transpose matrix $\mathbf{B}^\top$ acts as a numerical projection operator that computes the quadrature, ``pulling back'' the pointwise functional derivative from the integration grid into the lower-dimensional coefficient space $\mathbb{R}^d$. As defined in Section \ref{subsec:finite_obj}, the stiffness matrix $\mathbf{R}$ enforces the Sobolev roughness penalty exactly. This elegant matrix formulation ensures that the optimization remains computationally efficient while rigorously preserving the underlying continuous functional geometry. 

To optimize performance while maintaining this efficiency, we implement a first-order update rule. At each iteration $k$, the coefficients are updated as follows:
$$
    \mathbf{c}_{k+1} = \mathbf{c}_k - \alpha_k \nabla_{\mathbf{c}} J(\mathbf{c}_k)
$$
The step size $\alpha_k$ follows a predefined schedule that avoids the overhead of internal search loops. This approach efficiently mirrors the geometric ``flow'' on the warping manifold. The stability of the descent is further bolstered by the Sobolev penalty gradient $2\lambda \mathbf{R} \mathbf{c}$, which smooths the objective functional and prevents the optimizer from encountering the singular regions of high curvature typically associated with the pinching effect in standard unregularized registration.

\subsubsection{Functional Derivatives for the Four Frameworks}

The numerical optimization of the warping function relies on the exact computation of the continuous gradient with respect to the CLR representation. Based on the rigorous adjoint derivations provided in \ref{app:derivatives}, we summarize the final functional derivatives for each of the four registration frameworks below.

Given $\gamma = \Phi^{-1}(\psi)$, let $r(t) = f(t) - g(\gamma(t))$ represent the forward alignment residual. The Fr\'{e}chet derivatives $\frac{\delta \mathcal{D}_i}{\delta \psi}$, serving as the kernels for the finite-dimensional gradient assembly, are expressed using the centered adjoint kernel $\mathcal{K}_i(t)$ as follows:
$$
    \frac{\delta \mathcal{D}_i}{\delta \psi}(t) = \gamma'(t) \left( \mathcal{K}_i(t) - \int_{0}^{1} \mathcal{K}_i(s) \gamma'(s) ds \right)
$$

The specific adjoint kernels for each framework, which aggregate positional and Jacobian-based forces, are defined below:

\begin{enumerate}
    \item \textbf{Standard $\mathbb{L}^2$ (Method 1):}
    The kernel is purely positional, representing the cumulative impact of a local perturbation on the future alignment of the signal peaks.
    $$
        \mathcal{K}_1(t) = \int_{t}^{1} -2 r(s) g'(\gamma(s)) ds
    $$
    \item \textbf{Symmetric $\mathbb{L}^2$ (Method 2):}
    The kernel balances the positional forces of the forward and inverse warping errors with an instantaneous pressure term derived from the change-of-variables Jacobian.
    $$
        \mathcal{K}_2(t) = \int_{t}^{1} \left[ - r(s) g'(\gamma(s)) (1 + \gamma'(s)) \right] ds + r(t)^2
    $$

    \item \textbf{Isometry (Method 3):}
    The kernel incorporates the density-preserving residual $f(t) - g(\gamma(t))\sqrt{\gamma'(t)}$. Here, the Jacobian acts as a local amplitude scaler.
    $$
        \mathcal{K}_3(t) = \int_{t}^{1} -2  (f(s) - g(\gamma(s))\sqrt{\gamma'(s)}) g'(\gamma(s)) \sqrt{\gamma'(s)} ds - \frac{ f(t) - g(\gamma(t))\sqrt{\gamma'(t)} g(\gamma(t))}{\sqrt{\gamma'(t)}}
    $$

    \item \textbf{Jacobian-Weighted $\mathbb{L}^2$ (Method 4):}
    Utilizing the square-root of the Jacobian as a weighting measure, this kernel modulates the penalty of the Jacobian residual while weighting the positional force by the local square-root velocity.
    $$
        \mathcal{K}_4(t) = \int_{t}^{1} -2 r(s) g'(\gamma(s)) \sqrt{\gamma'(s)} ds + \frac{r(t)^2}{2\sqrt{\gamma'(t)}}
    $$
\end{enumerate}

\subsubsection{Computational Algorithm and Its Complexity}

The registration process is unified under a single algorithmic framework, detailed in Algorithm~\ref{alg:unified_reg}. The modularity of the continuous Fr\'{e}chet derivatives, encapsulated in the adjoint kernels $\mathcal{K}_i(t)$, allows for an identical optimization loop across all four methods, with only the specific metric choice altering Step 3.

\begin{algorithm}[H]
{\small
\caption{Unified Finite-Basis Sobolev Registration}
\label{alg:unified_reg}
\begin{algorithmic}[1]
\STATE \textbf{Initialize:} Coefficient vector $\mathbf{c} = \mathbf{0} \in \mathbb{R}^{d}$, learning rate $\alpha_k$ \hfill $\triangleright \; O(d)$
\STATE \textbf{Precompute:} Sobolev stiffness matrix $\mathbf{R}$ (banded) and Basis Matrix $\mathbf{B}$ \hfill $\triangleright \; O(d + N \cdot d)$
\WHILE{not converged}
    \STATE \textbf{1. Generate Warp:} Compute the CLR representation $\psi_{\mathbf{c}}(t) = \sum_{j=1}^{d} c_j \phi_j(t)$ \hfill $\triangleright \; O(N \cdot d)$
    \STATE $\gamma'(t) \gets \exp(\psi_{\mathbf{c}}(t)) / \int_0^1 \exp(\psi_{\mathbf{c}}(s)) ds$ \hfill $\triangleright \; O(N)$
    \STATE $\gamma(t) \gets \int_0^t \gamma'(s) ds$ \hfill $\triangleright \; O(N)$
    \STATE \textbf{2. Interpolation:} Evaluate $g(\gamma(t))$ and $g'(\gamma(t))$ via cubic splines on quadrature grid $N$ \hfill $\triangleright \; O(N)$
    \STATE \textbf{3. Gradient Assembly:} 
    \STATE Compute adjoint kernel $\mathcal{K}_i(t)$ (Case-specific via numerical integration) \hfill $\triangleright \; O(N)$
    \STATE $\mathbf{G}(t) \gets \gamma'(t) \left( \mathcal{K}_i(t) - \int_0^1 \mathcal{K}_i(s)\gamma'(s) ds \right)$ \hfill $\triangleright \; O(N)$
    \STATE \textbf{4. Subspace Projection:}
    \STATE $\nabla_{\mathbf{c}} J \gets \mathbf{B}^\top \mathbf{G} + 2\lambda \mathbf{R} \mathbf{c}$ \hfill $\triangleright \; O(N \cdot d)$
    \STATE \textbf{5. Step Update:} 
    \STATE $\mathbf{c} \gets \mathbf{c} - \alpha_k \nabla_{\mathbf{c}} J$ \hfill $\triangleright \; O(d)$
\ENDWHILE
\RETURN Warping function $\gamma(t)$ and optimal coefficients $\mathbf{\hat{c}}$
\end{algorithmic}
}
\end{algorithm}

A significant advantage of the proposed algorithm is the attainment of \textit{computational parity} across diverse geometric constraints. Our kernel formulations successfully unify the per-iteration cost across all four methods. As annotated step-by-step in Algorithm~\ref{alg:unified_reg}, the complexity of each operation is determined entirely by the interplay between the numerical quadrature grid size $N$ and the subspace dimension $d$, resulting in the overall costs summarized in Table~\ref{tab:complexity}.

\begin{table}[ht]
\centering
\caption{Computational Complexity of Functional Registration Components}
\label{tab:complexity}
\small
\renewcommand{\arraystretch}{1.5}
\begin{tabular}{|l|c|}
\hline
\textbf{Component} & \textbf{Complexity Order} \\ \hline
\textbf{Mismatch Term (Methods 1--4)} & $O(N \cdot d)$ \\ \hline
\textbf{Sobolev Penalty} & $O(d)$ \\ \hline
\end{tabular}
\end{table}

Based on this result, we have the following observations:
\begin{itemize}
    \item \textit{Overall Linear Scaling:} With the Sobolev penalty requiring only a linear $O(d)$ operation, the total per-iteration complexity across all four frameworks scales strictly as $O(N \cdot d)$. Since $d \ll N$ (e.g., $d = 40$ vs. $N = 1000$), the computational cost is overwhelmingly dominated by the numerical projection of the functional gradient onto the basis subspace. This strict linear scaling ensures that high-resolution signals can be processed efficiently, entirely circumventing the cubic or quadratic computational blow-up typically associated with traditional non-parametric registration methods or dynamic programming.

    \item \textit{Decoupled and Sparse Regularization:} The Sobolev penalty $\|\psi_{\mathbf{c}}\|_{\mathbb{H}}^2$ is computed entirely within the compressed coefficient space. Because the B-spline basis functions possess local support, the pre-computed stiffness matrix $\mathbf{R}$ is strictly banded. Consequently, the regularization step operates in $O(d)$ time and remains completely independent of the grid density $N$. This decoupling allows for extremely dense sampling of the signal without increasing the cost of enforcing smoothness on the warping field.
%
\end{itemize}

\subsection{Numerical Experiments and Comparative Analysis}

Following the development of the Sobolev-regularized registration framework, we evaluate the four proposed mismatch functionals to assess their structural robustness and ability to decouple phase from amplitude variability. 

To simulate a relatively challenging registration scenario, we design an ``inverse seesaw'' challenge characterized by topological conflict.  We define a discrete time domain $t \in [0, 1]$ with N = 1000 points. The source signal $g(t)$ is constructed as a mixture of a short-wide Gaussian and a tall-narrow Gaussian:
$$
    g(t) = 0.6 \exp\left(-\frac{(t-0.3)^2}{2(0.10)^2}\right) + 1.5 \exp\left(-\frac{(t-0.7)^2}{2(0.04)^2}\right)
$$
The target signal $f(t)$ is generated by taking a latent signal with inverse structural properties (a tall-narrow peak and a short-wide peak) and applying a latent ground-truth warping function $\gamma_{gt}(t)$:
$$
    f(t) = \left[ 1.4 \exp\left(-\frac{(t-0.3)^2}{2(0.04)^2}\right) + 0.5 \exp\left(-\frac{(t-0.7)^2}{2(0.10)^2}\right) \right] \circ \gamma_{gt}(t)
$$
where the warping function $\gamma_{gt}$ introduces significant non-linear phase distortion:
$$
    \gamma_{gt}(t) = t + 0.22\sin(\pi t)
$$

This configuration ensures that $f(t)$ and $g \circ \gamma(t)$ have fundamentally different local $\mathbb{L}^2$ norms. Finally, additive white Gaussian noise $\epsilon \sim \mathcal{N}(0, 0.04^2)$ is added to both signals. This setup forces the optimizer to distinguish between legitimate temporal displacement (phase) and structural height/width differences (amplitude). 

\begin{figure}[htbp]
\centering
\includegraphics[width=0.9\textwidth]{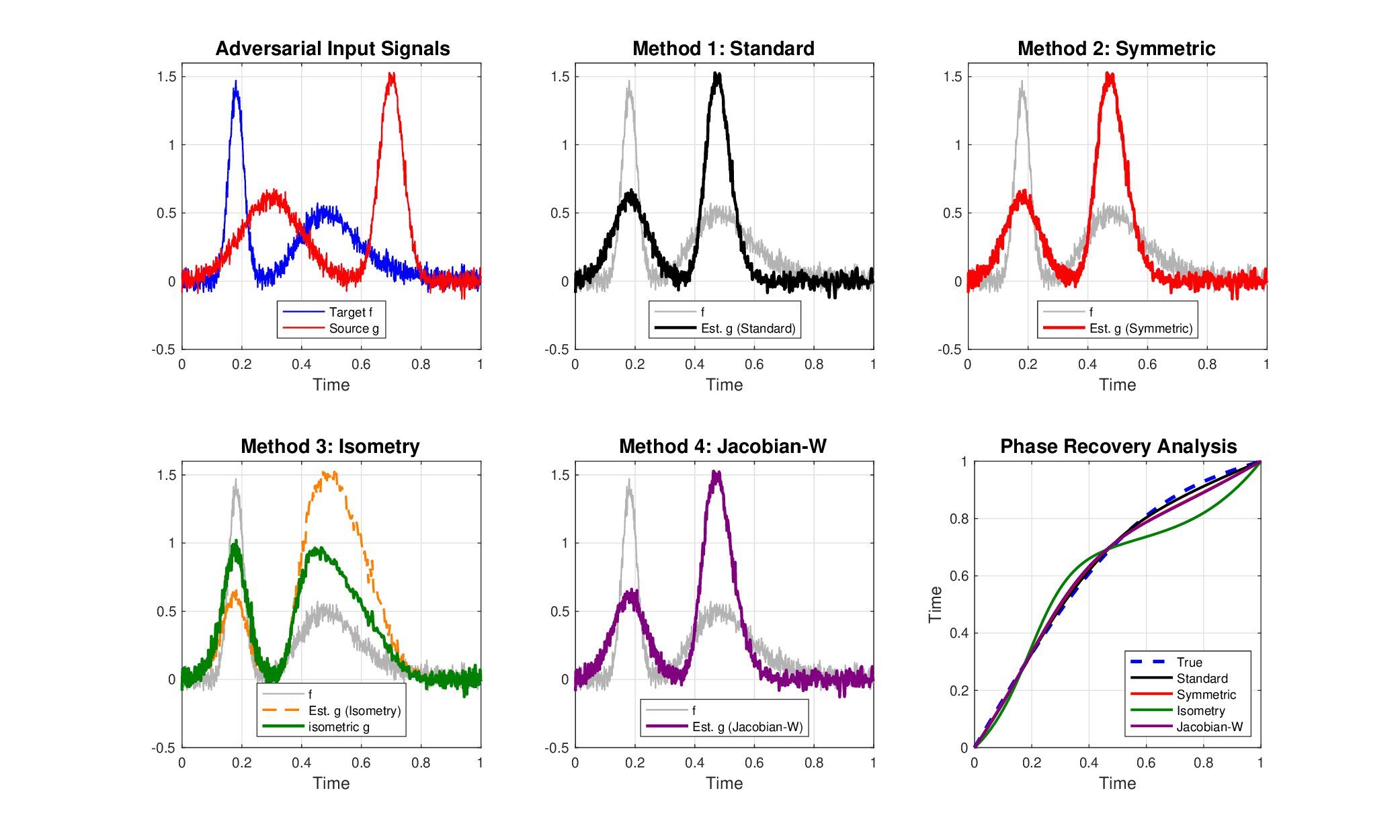}
\caption{Comparative analysis of registration objectives. The top-left panel shows the adversarial noisy inputs with inverted amplitudes. The top-middle to bottom-middle panels show the registered source signals for Methods 1-4, respectively. The bottom-right panel reveals that Method 1 (black) tracks the ground-truth $\gamma_{gt}$ (blue dashed) most accurately. Methods 2 (red) and 4 (purple) also demonstrate high fidelity. In contrast, Method 3 (green) shows the largest deviation.}
\label{fig:simu_results}
\end{figure}


We utilize a finite basis of $d=40$ cubic B-splines with a strictly defined zero-mean Sobolev regularization weight of $\lambda = 1 \times 10^{-5}$ and a constant learning rate $\alpha = 0.05$.
The primary finding of this experiment, illustrated in Figure \ref{fig:simu_results} and summarized in Table \ref{tab:comparison}, is that Methods 1, 2, and 4 all provide accurate phase recovery, while Method 3 exhibits a significant structural bias.

As shown in the Phase Recovery Analysis (Figure \ref{fig:simu_results}, bottom-right panel), Method 1 (black line) remains visually most faithful to the ground-truth $\gamma_{gt}$.  This suggests that when supported by the full $\mathbb{H}$-norm penalty, the standard $\mathbb{L}^2$ framework is highly effective at ignoring adversarial amplitude differences. The penalty effectively pre-empts the pinching singularities, allowing the optimizer to focus on structural alignment without the need for additional geometric weighting. Methods 2 (red) and 4 (purple) also perform exceptionally well, tracking the ground truth closely throughout the time domain. 

In the results for Method 3 (bottom-left panel), we display both the warped source $g \circ \gamma$ and its isometric transformation $(g \circ \gamma)\sqrt{\dot{\gamma}}$. It is highly instructive to observe how the peak widths and heights are modified to match the target. By utilizing the $\sqrt{\dot{\gamma}}$ term, Method 3 easily scales the signal's vertical range, increasing the height of the short peak by locally compressing its width. Because Method 3 can compensate for height via the derivative term, it exploits the warping function to minimize vertical residuals. While this achieves an impressive visual fit, the phase recovery analysis confirms that it introduces a significant structural bias.

\begin{table}[htbp]
\centering
\caption{Quantitative Performance Summary of Registration Objectives}
\label{tab:comparison}
\begin{tabular}{lcccc}
\hline
\textbf{Mismatch Term} & \textbf{1. Standard} & \textbf{2. Symmetric} & \textbf{3. Isometry} & \textbf{4. Jacobian-W} \\ \hline
\textbf{Phase Recovery} & High & High & Low (Biased) & High \\
\textbf{CLR Dist ($L^2$)} & 0.1281 & 0.2490 & 0.6911 & 0.2660 \\
\textbf{$\mathbb{H}$-Norm Dist} & 23.08 & 20.30 & 95.69 & 24.99 \\ \hline
\end{tabular}
\end{table}

To rigorously confirm these observations, we calculate the quantitative distances between the estimated warping functions and the ground truth. We evaluate the standard $\mathbb L^2$ distance in the CLR space, as well as the strict theoretical $\mathbb{H}$-norm distance, which measures the deviations in the first and second derivatives.  Table \ref{tab:comparison} presents these metrics. The quantitative results demonstrate that Methods 1, 2, and 4 achieve low CLR distances, validating their consistency. The Symmetric formulation (Method 2) achieves the lowest $\mathbb{H}$-norm distance (20.30), indicating superior smoothness and derivative fidelity, while the Standard formulation (Method 1) achieves the absolute tightest CLR distance (0.1281). Method 3, however, exhibits a massive $\mathbb{H}$-norm error (95.69) and a severe CLR divergence (0.6911), quantitatively confirming that the isometry-preserving mismatch alters the true phase.

\subsection{Consistency Analysis: Robust Phase Recovery Under Additive Noise}

Following the asymptotic theory presented in Section \ref{subsec:consis}, we provide a numerical demonstration of phase recovery under adversarial conditions. To examine the robustness of the estimation process, we introduce significant additive noise and a global intensity mismatch. 

\begin{figure}[htbp]
    \centering
    \includegraphics[width=0.9\textwidth]{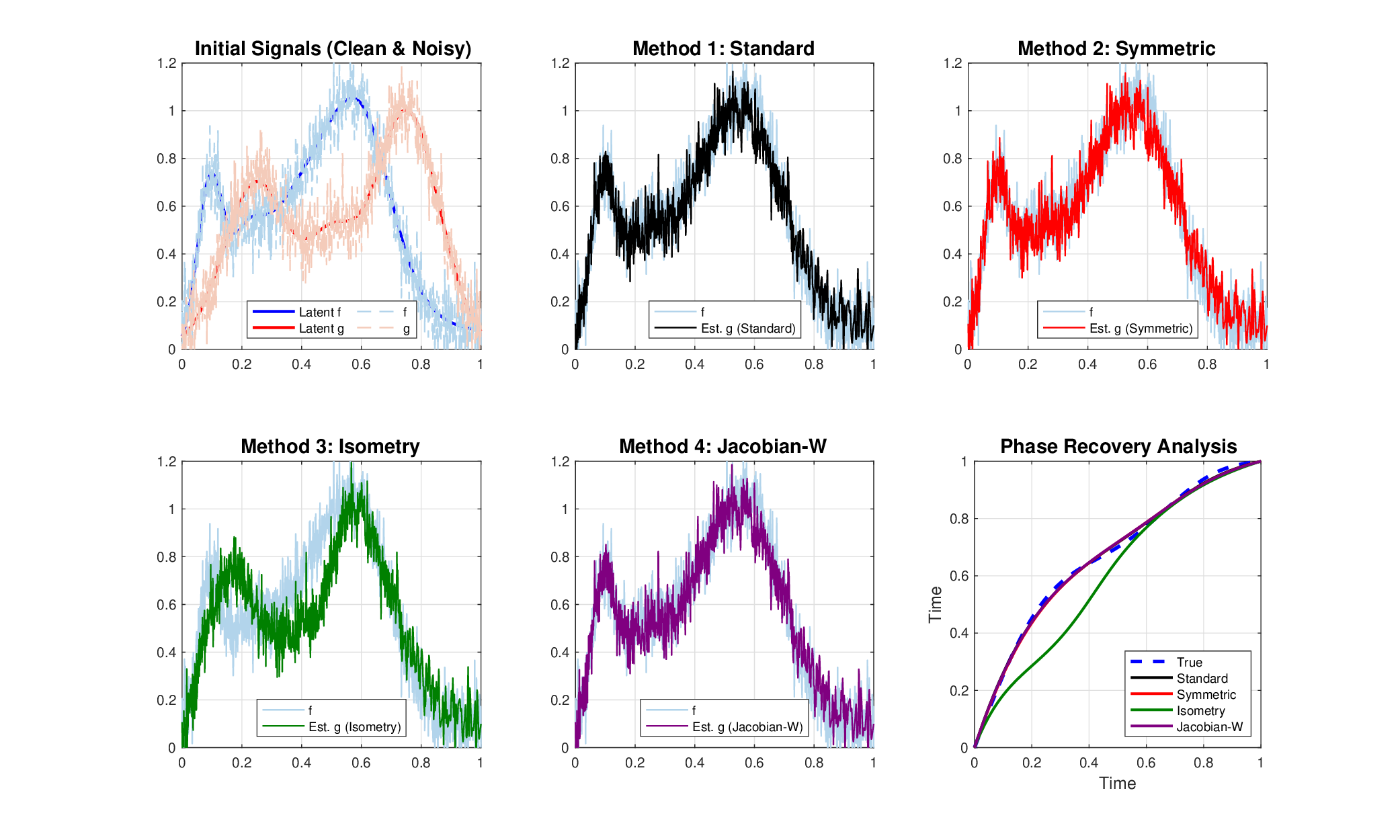}
    \caption{Consistency analysis and robustness test. Top-left: Initial latent and observed signals ($f$ in blue, $g$ in red). Top-middle to bottom-middle: Alignment results for Methods 1-4, respectively. Bottom-right: Phase recovery analysis showing the recovered warping functions against the true $\gamma_{gt}$ (dashed blue). Methods 1 (black), 2 (red), and 4 (purple) track the ground truth with high fidelity, while Method 3 (green) exhibits a visible phase bias.}
    \label{fig:consistency_analysis}
\end{figure}

\textit{Experimental Design:} The latent source signal $g_{latent}(t)$ is modeled as a mixture of three overlapping Gaussian kernels:
$$
    g_{latent}(t) = \sum_{i=1}^3 A_i \exp\left(-\frac{(t-\mu_i)^2}{2\sigma_i^2}\right)
$$
with amplitudes $\mathbf{A} = [0.7, 0.4, 1.0]$, centers $\boldsymbol{\mu} = [0.25, 0.5, 0.75]$, and scales $\boldsymbol{\sigma} = [0.11, 0.08, 0.11]$. The latent target $f_{latent}(t)$ is generated via a ground-truth warping $\gamma_{gt}(t)$ and a 5\% intensity scaling for display purposes:
$$
    f_{latent}(t) = 1.05 \cdot (g_{latent} \circ \gamma_{gt})(t)
$$
where the warping function is defined by the harmonic series:
$$
    \gamma_{gt}(t) = t + 0.25\sin(\pi t) + 0.06\sin(2\pi t) + 0.05\sin(3\pi t)
$$
To evaluate robustness, the final observed signals $f(t)$ and $g(t)$ are corrupted with independent additive white Gaussian noise $\epsilon_f(t), \epsilon_g(t) \sim \mathcal{N}(0, 0.08^2)$:
$$
    f(t) = f_{latent}(t) + \epsilon_f(t), \quad g(t) = g_{latent}(t) + \epsilon_g(t)
$$
We utilize a finite basis of $d=20$ cubic B-splines with a strictly defined zero-mean Sobolev regularization weight of $\lambda = 4 \times 10^{-6}$ and a learning rate of $\alpha = 0.8$.

The numerical results, illustrated in Figure \ref{fig:consistency_analysis}, support the theoretical prediction that standard, symmetric, and Jacobian-weighted objectives are robust to noise, whereas the isometric metric suffers from structural bias. To rigorously quantify this consistency, we measure the recovery error in the CLR space, calculating both the standard $\mathbb{L}^2$ distance between the estimated and true centered log-derivatives and the strict $\mathbb{H}$-norm distance, which evaluates the structural fidelity of the first and second derivatives.

As evidenced by the quantitative results in Table \ref{tab:nrr_stats}, Methods 1, 2, and 4 yield strong phase recovery, but with an important theoretical distinction. Method 1 achieves the lowest pointwise CLR distance (0.2238). However, when evaluated under the strict $\mathbb{H}$-norm, which focuses on the continuous differentiability and smoothness of the warp, the Symmetric formulation (Method 2) and Jacobian-Weighted formulation (Method 4) achieve strictly superior structural fidelity (76.34 and 76.77, respectively, compared to Method 1's 82.06). This indicates that despite the intensity mismatch and high noise level, the geometrically aware frameworks effectively identify a more topologically consistent and smoother underlying velocity field. 
In contrast, Method 3 serves as a counter-example, producing a CLR distance (0.5202) and $\mathbb{H}$-norm distance (115.41) significantly higher than the other frameworks.

\begin{table}[htbp]
\centering
\caption{Consistency Statistics: Quantitative Recovery Error}
\label{tab:nrr_stats}
\small
\begin{tabular}{lcccc}
\hline
\textbf{Mismatch Term}  & \textbf{1. Standard} & \textbf{2. Symmetric} & \textbf{3. Isometry} & \textbf{4. Jacobian-W} \\ \hline
\textbf{Theoretical Consistency} & Satisfactory & Satisfactory & Biased & Satisfactory \\
\textbf{CLR Dist ($L^2$)}  & 0.2238 & 0.2463 & 0.5202 & 0.2461 \\
\textbf{$\mathbb{H}$-Norm Dist} & 82.06 & 76.34 & 115.41 & 76.77 \\ \hline
\end{tabular}
\end{table}

\subsection{Real Acoustic Data: Free Spoken Digit Dataset}
\label{sec:acoustic_data}

To evaluate the performance of the four Sobolev-regularized registration methods on real-world signals, we utilize the \textit{Free Spoken Digit Dataset} (FSDD) \citep{fsdd_github}. Often referred to as the ``audio version of MNIST,'' FSDD is an open-source repository consisting of recordings of spoken digits in \texttt{.wav} format, sampled at 8 kHz. The dataset contains 3000 recordings from six speakers, with 50 recordings per digit per speaker. For illustrative purposes in our pairwise registration task, we select two specific samples representing the digit ``0'':
\begin{itemize}
    \item \textbf{Target Signal ($f$):} \texttt{0\_george\_0.wav} (Speaker: George)
    \item \textbf{Source Signal ($g$):} \texttt{0\_jackson\_0.wav} (Speaker: Jackson)
\end{itemize}

These recordings were chosen because they exhibit natural non-linear temporal warping; George and Jackson possess distinct speaking rates and phonetic emphasis. The signals are pre-processed to extract their temporal envelopes via an RMS-based moving window, ensuring that the registration captures the structural alignment of the utterance rather than the high-frequency oscillation of the vocal carrier. 
The resulting envelopes are normalized to $[0, 1]$ and resampled to a uniform quadrature grid of $N=1000$ points for optimization.

\begin{figure}[htbp]
    \centering
    \includegraphics[width=0.8\textwidth]{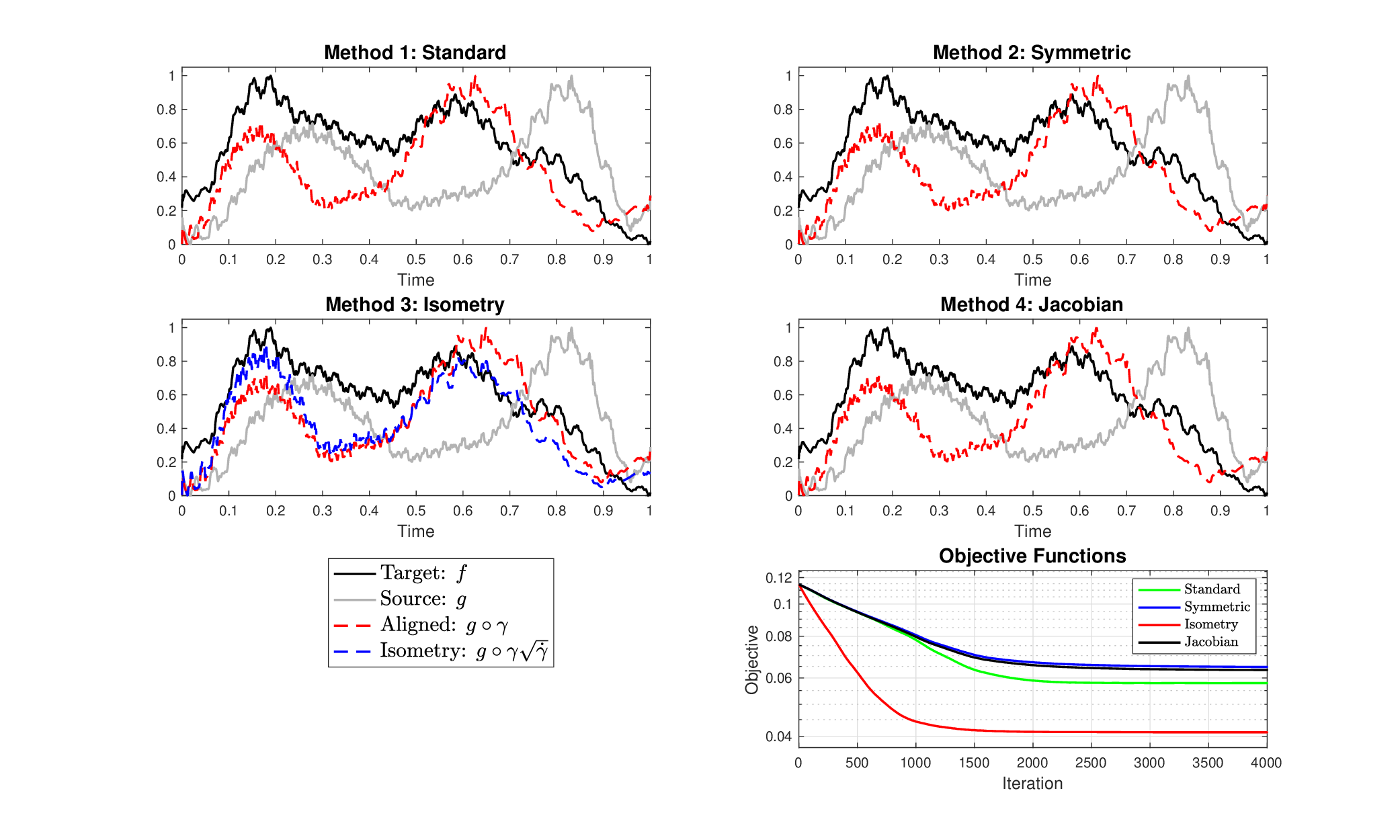}
\caption{Pairwise registration results for real acoustic data. The top row displays the alignment results for Method 1 (Standard $\mathbb{L}^2$) and Method 2 (Symmetric $\mathbb{L}^2$), while the middle row presents Method 3 (Isometry) and Method 4 (Jacobian-Weighted). A common legend is provided in the bottom-left panel, identifying the target signal $f$ (solid black), the source signal $g$ (solid gray), and the aligned signal $g \circ \gamma$ (dashed red). For Method 3, the isometrically aligned signal $(g \circ \gamma)\sqrt{\gamma'}$ is additionally displayed (dashed blue). Finally, the bottom-right panel illustrates the objective function convergence for all four methods to demonstrate numerical stability.}

    \label{fig:real_acoustic_registration}
\end{figure}

To ensure a fair comparison across the four methods, we employ a consistent set of computational parameters: a finite basis of $d=15$ cubic B-splines, governed by the strictly defined zero-mean Sobolev $\mathbb{H}$-norm penalty with a coefficient of $\lambda = 8 \times 10^{-5}$, and a gradient descent step size of $\alpha = 0.05$. Each optimization is initialized with the identity warping $\gamma_{id}$ and performed over 4000 iterations. 

The evolution of the objective functions is illustrated in the lower-right panel of Figure~\ref{fig:real_acoustic_registration}, with each formulation represented by a distinct color. Notably, Method 3 exhibits the fastest convergence, reaching the lowest objective value of 0.041, while the other three methods converge to 0.058, 0.064, and 0.063, respectively. However, it is important to note that the data mismatch terms for Methods 2, 3, and 4 are mathematically symmetric; switching the roles of $f$ and $g$ and applying the optimal inverse warping $\gamma^{-1}$ yields nearly identical total objective values (due to the slight variation in the Sobolev penalty evaluated in the target domain). In contrast, evaluating the standard $\mathbb{L}^2$ mismatch (Method 1) under this spatial inversion shifts the final value significantly to 0.073, empirically highlighting the inherent asymmetry of that baseline framework.

The alignment results for the four evaluated methods are presented in the top two rows of Figure~\ref{fig:real_acoustic_registration}. In each panel, the target $f$, the source $g$, and the warped signal $g \circ \gamma$ are displayed in distinct colors. Observations indicate that all methods yield visually accurate temporal alignments in the original signal space. However, because a pure warping process preserves the original signal values and cannot adjust peak heights, a clear discrepancy in amplitude naturally remains between the target and the time-warped source signals. 

Because Method 3 relies on an isometry-preserving data term, we also display its isometric representation $(g \circ \gamma) \sqrt{\gamma'}$ in the middle-left panel. This representation demonstrates a significantly closer visual match to the target $f$ than the raw warped signal. 
This result is consistent with the optimization objective of Method 3, which is specifically designed to minimize the $\mathbb{L}^2$ distance between the original target $f$ and the isometrically warped source $(g \circ \gamma)\sqrt{\gamma'}$. However, this comparison clearly illustrates how the isometric requirement exploits the Jacobian weight $\sqrt{\gamma'}$ to artificially introduce amplitude distortions into the warped signal to achieve vertical alignment. This real-world behavior confirms that isometry-based mismatch functionals sacrifice pure phase integrity to minimize amplitude residuals.

\section{Discussion}
\label{sec:discussion}

In this paper, we have introduced a novel, Sobolev-regularized deterministic framework for the robust pairwise registration of functional data. A fundamental challenge in functional data analysis is decoupling phase variability from amplitude variability, particularly when signals are corrupted by high levels of additive noise. Traditional derivative-based methods \citep{srivastava2011}, while mathematically elegant, often suffer in these noisy conditions due to the inherent instability of numerical differentiation. To address this, we propose operating entirely within the original function space, leveraging the CLR transform to map the strictly constrained manifold of diffeomorphic warping functions into an unconstrained space. By governing the optimization with a strictly defined zero-mean Sobolev $\mathbb{H}$-norm, our framework guarantees strictly monotonic, globally smooth alignments without requiring artificial boundary constraints.

Our functional registration framework is driven by a regularized objective comprising two essential components: a geometric data mismatch functional and a structural roughness penalty. Historically, regularizing directly in the original warping space \citep{tang2008pairwise, james2007curve, srivastava2016functional, kneip2008combining} has required expensive constrained optimization to prevent domain folding. While shifting the penalty to the unconstrained CLR space resolves this, existing zeroth-order \citep{ma2024}, first-order \citep{ramsay1998curve}, or purely second-order penalties remain geometrically and topologically deficient, permitting artificial identity bias, non-differentiable kinks, or unpenalized extreme stretching. To resolve these issues, our formulation necessitates a comprehensive Sobolev penalty ($\|\psi'\|_{\mathbb{L}^2}^2 + \|\psi''\|_{\mathbb{L}^2}^2$). This simultaneously establishes $\mathbb{H}$ as a complete Hilbert space and strictly bounds the maximum allowable relative stretch. By enforcing continuous differentiability, this unified approach intrinsically guarantees that the warping derivative remains uniformly bounded away from zero and infinity, mathematically ensuring the asymptotic consistency of the estimator.

Building on these topological guarantees, our framework systematically evaluates four distinct data mismatch formulations drawn from established functional alignment literature. As demonstrated by our results, there is a critical trade-off between pointwise alignment and structural fidelity. The Standard $\mathbb{L}^2$ baseline (Method 1) \citep{ramsay2005functional} achieves the tightest raw pointwise fits but suffers from an inherent ``target bias,'' yielding alignments that lack inverse consistency. Incorporating geometric symmetry, a principle highlighted in classical curve registration \citep{ramsay1998curve}, resolves this bias, though it introduces new nuances. Specifically, the Isometry formulation (Method 3) -- adapted from the widely used Square-Root Velocity Function (SRVF) framework \citep{srivastava2011, srivastava2016functional} -- achieves visually impressive alignments but exhibits a severe structural bias by forcing artificial amplitude distortion to minimize vertical residuals. In contrast, the Symmetric $\mathbb{L}^2$ (Method 2) \citep{tagare2009symmetric_theory} and Jacobian-Weighted (Method 4) \citep{wang1997alignment} frameworks treat both signals as active participants without enforcing artificial amplitude scaling. These geometrically aware formulations demonstrate exceptional robustness against topological traps, identifying true phase distortions and preserving the unbiased physical relationships between the curves.
%

Computationally, our approach relies on an optimization-based search to identify the optimal time-warping function. Unlike classical Fisher-Rao alignments that frequently utilize dynamic programming, a discrete technique that can be unstable and sensitive to additive noise, our deterministic framework enforces strict topological regularity through the continuous Sobolev penalty to naturally guarantee a smooth diffeomorphism. By projecting the infinite-dimensional problem onto a finite basis, the algorithm transforms the objective into a finite-dimensional optimization. This yields highly efficient, stable convergence with linear computational complexity, ensuring that the theoretical advantages of symmetric and inverse-consistent alignment do not compromise computational efficiency across any of the four evaluated methods.

While our framework offers robust deterministic performance, it inherently lacks a stochastic data model. Unlike Bayesian registration frameworks that sample from a posterior distribution of warping functions \citep{cheng2016bayesian, kurtek2017geometric, lu2017bayesian, Bharath2020landmark, matuk2021bayesian, tucker2021multimodal, ma2025new}, our method cannot quantify uncertainty in the alignment estimate. The optimization converges to a single, unique solution, meaning it does not provide multiple plausible warping possibilities or credible intervals. Thus, our framework strategically prioritizes computational efficiency, strong noise robustness, and strict topological guarantees over probabilistic inference, making it highly suitable for large-scale deterministic applications where processing speed and strict geometric constraints are paramount.

Our future work will proceed along three primary directions. First, to automate the registration pipeline across heterogeneous datasets, we will explore adaptive, data-driven selection criteria for the Sobolev penalty weight $\lambda$, potentially employing cross-validation or restricted maximum likelihood. Second, we aim to investigate Reproducing Kernel Hilbert Spaces (RKHS) as an alternative regularization penalty. Drawing upon the renowned Representer Theorem and commonly used Sieve estimation methods, RKHS offers profound theoretical elegance for handling analytical derivatives and establishing asymptotic properties. We will explore adapting these theoretical tools for continuous functional registration while maintaining computational efficiency. Finally, we plan to extend the current pairwise strategy to a multiple alignment setting, facilitating the estimation of a structural $K$-sample mean for functional observations.

\vspace{24pt}

\appendix

\noindent {\Large \bf Appendices}

\section{Proofs of Propositions in the Sobolev Space $\mathbb H$}
\label{app:sobolev_proof}

\subsection{Proof of Proposition \ref{prop:hilbert}}
\label{app:hilbert_proof}

\begin{proof}
To establish that $\mathbb{H}$ is a Hilbert space, we must show that $\langle \cdot, \cdot \rangle_{\mathbb{H}}$ constitutes a valid inner product and that $\mathbb{H}$ is complete with respect to the induced norm.

\textit{1. Validity of the Inner Product:} 
The functional $\langle \psi_1, \psi_2 \rangle_{\mathbb{H}}$ is trivially symmetric and bilinear. To prove positive definiteness, assume $\|\psi\|_{\mathbb{H}}^2 = 0$. By definition, this requires both $\|\psi'\|_{L^2}^2 = 0$ and $\|\psi''\|_{L^2}^2 = 0$. Since $\|\psi'\|_{L^2} = 0$, the derivative $\psi'$ is zero almost everywhere, implying $\psi(t) = c$ for some constant $c$. Because $\psi \in \mathbb{H}$, we have the zero-mean constraint $\int_0^1 \psi(t) dt = 0$. Substituting $\psi(t) = c$ yields $c = 0$, meaning $\psi(t) = 0$ everywhere. Thus, the inner product is strictly positive definite.

\textit{2. Completeness via Norm Equivalence:} 
The standard Sobolev space $W^{2,2}(I)$ is a known Hilbert space under its standard norm:
$$
    \|\psi\|_{W^{2,2}}^2 = \|\psi\|_{L^2}^2 + \|\psi'\|_{L^2}^2 + \|\psi''\|_{L^2}^2
$$
By definition, our induced norm is $\|\psi\|_{\mathbb{H}}^2 = \|\psi'\|_{L^2}^2 + \|\psi''\|_{L^2}^2$. Clearly, $\|\psi\|_{\mathbb{H}}^2 \leq \|\psi\|_{W^{2,2}}^2$. 

To bound it from the other direction, we use the intermediate result from the proof of Proposition \ref{prop:bound}, which established that for any $\psi \in \mathbb{H}$, $\|\psi\|_{L^2}^2 \leq \|\psi\|_\infty^2 \leq \|\psi'\|_{L^2}^2$. Substituting this into the standard norm gives:
$$
    \|\psi\|_{W^{2,2}}^2 \leq \|\psi'\|_{L^2}^2 + \|\psi'\|_{L^2}^2 + \|\psi''\|_{L^2}^2 \leq 2\|\psi'\|_{L^2}^2 + \|\psi''\|_{L^2}^2 \leq 2\|\psi\|_{\mathbb{H}}^2
$$
Therefore, $\frac{1}{2}\|\psi\|_{W^{2,2}}^2 \leq \|\psi\|_{\mathbb{H}}^2 \leq \|\psi\|_{W^{2,2}}^2$, proving that our induced norm $\|\cdot\|_{\mathbb{H}}$ is topologically equivalent to the standard $W^{2,2}(I)$ norm on this subspace.

\textit{3. Closed Subspace:}
Consider the linear functional $T: W^{2,2}(I) \to \mathbb{R}$ defined by $T(\psi) = \int_0^1 \psi(t) dt$. Because $|T(\psi)| \leq \|\psi\|_{L^1} \leq \|\psi\|_{L^2} \leq \|\psi\|_{W^{2,2}}$, $T$ is a bounded (and thus continuous) linear operator. The space $\mathbb{H}$ is exactly the kernel of this operator: $\mathbb{H} = \ker(T)$. 

Because the kernel of any continuous linear functional is closed, $\mathbb{H}$ is a closed subspace of the complete space $W^{2,2}(I)$. A closed subspace of a complete space is itself complete. Because $\mathbb{H}$ is complete under the equivalent norm $\|\cdot\|_{\mathbb{H}}$ and its geometry is induced by a valid inner product, $\mathbb{H}$ is a Hilbert space.
\end{proof}

\subsection{Proof of Proposition \ref{prop:bound}}
\label{app:bound_proof}

\begin{proof}
Let $\psi \in \mathbb{H}$. Because $\psi$ is continuous on the compact interval $I=[0,1]$ and has a zero mean ($\int_0^1 \psi(t) dt = 0$), the Mean Value Theorem for Integrals guarantees the existence of a point $t_0 \in [0,1]$ such that $\psi(t_0) = 0$.

By the Fundamental Theorem of Calculus, for any $t \in [0,1]$:
$
    \psi(t) = \int_{t_0}^t \psi'(s) ds.
$
Taking the absolute value and applying the Cauchy-Schwarz inequality yields:
$$
    |\psi(t)| \leq \int_0^1 |\psi'(s)| \cdot 1 ds \leq \left( \int_0^1 |\psi'(s)|^2 ds \right)^{1/2} \left( \int_0^1 1^2 ds \right)^{1/2} = \|\psi'\|_{L^2}
$$
Taking the supremum over all $t \in [0,1]$ gives $\|\psi\|_\infty \leq \|\psi'\|_{L^2}$. By the definition of the inner product on $\mathbb{H}$, we know $\|\psi'\|_{L^2}^2 \leq \|\psi'\|_{L^2}^2 + \|\psi''\|_{L^2}^2 = \|\psi\|_{\mathbb{H}}^2$. Taking the square root establishes the first bounded relationship: 
$$
    \|\psi\|_\infty \leq \|\psi'\|_{L^2} \leq \|\psi\|_{\mathbb{H}}
$$

To bound the derivative $\psi'$, we again apply the Mean Value Theorem for Integrals. Because $\psi'$ is continuous, there exists a point $t_1 \in [0,1]$ such that $\psi'(t_1) = \int_0^1 \psi'(s) ds$. 

Applying the Fundamental Theorem of Calculus to $\psi'$ from $t_1$ to $t$:
$$
    \psi'(t) = \psi'(t_1) + \int_{t_1}^t \psi''(s) ds = \int_0^1 \psi'(s) ds + \int_{t_1}^t \psi''(s) ds
$$
Taking the absolute value and extending the integration bounds over the entire interval $[0,1]$:
$$
    |\psi'(t)| \leq \int_0^1 |\psi'(s)| ds + \int_0^1 |\psi''(s)| ds
$$
Applying the Cauchy-Schwarz inequality to each term ($L^1$ norm $\leq L^2$ norm on a unit interval), we get:
$$
    |\psi'(t)| \leq \|\psi'\|_{L^2} + \|\psi''\|_{L^2}
$$
Using the fundamental algebraic inequality $(a+b)^2 \leq 2(a^2 + b^2)$ for any real numbers $a, b$, we have:
$$
    \|\psi'\|_{L^2} + \|\psi''\|_{L^2} \leq \sqrt{2 \left( \|\psi'\|_{L^2}^2 + \|\psi''\|_{L^2}^2 \right)} = \sqrt{2} \|\psi\|_{\mathbb{H}}
$$
Taking the supremum over $t$ yields the final uniform bound: 
$$
    \|\psi'\|_\infty \leq \sqrt{2} \|\psi\|_{\mathbb{H}}
$$
This completes the proof.
\end{proof}

\section{Proof of Proposition \ref{prop:uni} on Uniqueness of Self-Alignment}
\label{app:uni_proof}

    \begin{proof}
    Suppose $f = f \circ \gamma$. Since $f$ is continuous and piecewise strictly monotonic, any mapping $\gamma$ that preserves the signal value must map each monotonic segment to itself. Within these intervals, $f$ is injective; therefore, $f(t) = f(\gamma(t)) \implies t = \gamma(t)$. Combined with the strictly increasing nature of $\gamma$ and the boundary conditions $\gamma(0)=0, \gamma(1)=1$, the identity map is the only continuous solution.
    \end{proof}

\section{Detailed Proofs of Existence for All Mismatch Formulations}
\label{app:existence_proof}

We provide the formal proofs for the existence of optimal warping fields across all four mismatch formulations within our single registration framework.

\subsection{Existence Proof for Method 1: Standard $\mathbb{L}^2$ Mismatch}
\label{app:existence_proof1}

\begin{proof}
    Since the mismatch term $\mathcal{D}_1(\psi) \ge 0$ and the penalty term $\lambda \|\psi\|_{\mathbb{H}}^2 \ge 0$, the functional is bounded below. Let $M = \inf_{\psi \in \mathbb{H}} \mathcal{O}_1(\psi)$. By the Poincaré-Wirtinger inequality, the zero-mean constraint on $\mathbb{H}$ ensures that the Sobolev seminorm $\|\psi\|_{\mathbb{H}}$ (incorporating both $\psi'$ and $\psi''$ terms) is equivalent to the full $W^{2,2}(I)$ norm. Thus, as $\|\psi\|_{\mathbb{H}} \to \infty$, the penalty term dominates, and $\mathcal{O}_1(\psi) \to \infty$, establishing coercivity.

    Let $\{\psi_n\} \subset \mathbb{H}$ be a minimizing sequence such that $\mathcal{O}_1(\psi_n) \to M$. Coercivity implies that $\{\psi_n\}$ is bounded in $\mathbb{H}$. By the reflexivity of Sobolev spaces, there exists a subsequence $\{\psi_{n_k}\}$ such that $\psi_{n_k} \rightharpoonup \psi^*$ weakly in $\mathbb{H}$. As $\mathbb{H} \subset W^{2,2}(I) \hookrightarrow C^1(I)$ compactly, the weak convergence in $\mathbb{H}$ implies strong convergence in the $C^1(I)$ norm:
    $$
        \|\psi_{n_k} - \psi^*\|_\infty \to 0, \quad \|\psi'_{n_k} - (\psi^*)'\|_\infty \to 0
    $$
    
    The uniform convergence $\psi_{n_k} \to \psi^*$ implies the uniform convergence of the warping functions $\gamma_{n_k} \to \gamma^*$. In our CLR framework, the derivative is defined as $\gamma'(t) = \exp(\psi(t)) / \int_0^1 \exp(\psi(u)) \, du$. Since $\psi_{n_k} \to \psi^*$ uniformly and $\psi^* \in C^1(I)$, the limit derivative $(\gamma^*)'$ is continuous and strictly positive. This ensures that the limit $\gamma^*$ is a strictly increasing diffeomorphism, structurally preventing any pinching.

The penalty term $\lambda \|\psi\|_{\mathbb{H}}^2$ is a squared norm scaled by a positive constant and is thus weakly lower semi-continuous. Since $f, g \in C^0(I)$, the uniform convergence $\gamma_{n_k} \to \gamma^*$ implies the entire integrand $(f - g \circ \gamma_{n_k})^2$ converges uniformly to $(f - g \circ \gamma^*)^2$. This uniform convergence on the compact interval $I$ allows us to pass the limit directly inside the integral:
$$
    \lim_{k \to \infty} \mathcal{D}_1(\psi_{n_k}) = \mathcal{D}_1(\psi^*)
$$
    
    Combining these properties, we have:
    $$
        \mathcal{O}_1(\psi^*) \le \lim_{k \to \infty} \mathcal{D}_1(\psi_{n_k}) + \lambda \liminf_{k \to \infty} \|\psi_{n_k}\|_{\mathbb{H}}^2 = M
    $$
    Since $M$ is the infimum, it follows that $\mathcal{O}_1(\psi^*) = M$, proving that $\psi^*$ is a global minimizer in $\mathbb{H}$.
\end{proof}

\subsection{Existence Proof for Method 2: Symmetric $\mathbb{L}^2$ Mismatch}
\label{app:existence_proof2}

\begin{proof}
    Since the mismatch term $\mathcal{D}_2(\psi) \ge 0$, the penalty term $\lambda \|\psi\|_{\mathbb{H}}^2$ dominates as $\|\psi\|_{\mathbb{H}} \to \infty$. Thus, the functional $\mathcal{O}_2(\psi)$ is coercive. Let $M = \inf_{\psi \in \mathbb{H}} \mathcal{O}_2(\psi)$, and let $\{\psi_n\} \subset \mathbb{H}$ be a minimizing sequence such that $\mathcal{O}_2(\psi_n) \to M$.

    By coercivity, the sequence $\{\psi_n\}$ is bounded in $\mathbb{H}$. As established in the previous proof, the reflexivity of $\mathbb{H}$ and the compact embedding $\mathbb{H} \hookrightarrow C^1(I)$ allow us to extract a subsequence such that $\psi_{n_k} \rightharpoonup \psi^*$ weakly in $\mathbb{H}$ and $\psi_{n_k} \to \psi^*$ strongly in $C^1(I)$. This yields uniform convergence for both the warping function and its derivative:
    $$
        \|\gamma_{n_k} - \gamma^*\|_\infty \to 0, \quad \|\gamma'_{n_k} - (\gamma^*)'\|_\infty \to 0
    $$

    Assuming the signals $f, g \in C^0(I)$, the uniform convergence $\gamma_{n_k} \to \gamma^*$ ensures that $(f - g \circ \gamma_{n_k})^2 \to (f - g \circ \gamma^*)^2$ uniformly. Simultaneously, the weighting term $\frac{1+\gamma'_{n_k}}{2}$ converges uniformly to $\frac{1+(\gamma^*)'}{2}$. Because the product of uniformly convergent sequences on a compact domain is also uniformly convergent, the entire weighted integrand converges uniformly:
    $$
        (f(t) - g(\gamma_{n_k}(t)))^2 \left( \frac{1+\gamma'_{n_k}(t)}{2} \right) \to (f(t) - g(\gamma^*(t)))^2 \left( \frac{1+(\gamma^*)'(t)}{2} \right)
    $$

    Uniform convergence on the compact interval $I$ allows the limit to pass inside the integral:
    $$
        \lim_{k \to \infty} \mathcal{D}_2(\psi_{n_k}) = \mathcal{D}_2(\psi^*)
    $$

    Utilizing the weak lower semi-continuity of the Sobolev norm in $\mathbb{H}$:
    $$
        \mathcal{O}_2(\psi^*) \leq \lim_{k \to \infty} \mathcal{D}_2(\psi_{n_k}) + \lambda \liminf_{k \to \infty} \|\psi_{n_k}\|_{\mathbb{H}}^2 = M
    $$
    Since $M$ is the infimum, it follows that $\mathcal{O}_2(\psi^*) = M$. Thus, $\psi^*$ is a global minimizer for the Symmetric $\mathbb{L}^2$ functional.
\end{proof}

\subsection{Existence Proof for Method 3: Isometry ($\mathbb{L}^2$-Preserving) Mismatch}
\label{app:existence_proof3}

\begin{proof}
    Since the mismatch term $\mathcal{D}_3(\psi) \ge 0$, the penalty term $\lambda \|\psi\|_{\mathbb{H}}^2$ drives the functional to infinity as $\|\psi\|_{\mathbb{H}} \to \infty$. By this coercivity, any minimizing sequence $\{\psi_n\} \subset \mathbb{H}$ is bounded in $\mathbb{H}$. By the reflexivity of Hilbert spaces, there exists a subsequence such that $\psi_{n_k} \rightharpoonup \psi^*$ weakly in $\mathbb{H}$.
    
    As established in the previous proofs, the compact embedding $\mathbb{H} \hookrightarrow C^1(I)$ ensures that weak convergence in $\mathbb{H}$ implies strong convergence in $C^1(I)$. Consequently, we have uniform convergence for the warping function and its derivative:
    $$
        \|\gamma_{n_k} - \gamma^*\|_\infty \to 0, \quad \|\gamma'_{n_k} - (\gamma^*)'\|_\infty \to 0
    $$

    Following the logic for Method 1, the uniform convergence of $\gamma_{n_k}$ and the continuity of $g$ imply that $g \circ \gamma_{n_k} \to g \circ \gamma^*$ uniformly. In our CLR framework, the limit derivative $(\gamma^*)' > 0$ and is continuous. Because the square root function is uniformly continuous on $[0, \infty)$, the uniform convergence of the derivatives guarantees that the square-root operator also converges uniformly: $\sqrt{\gamma'_{n_k}} \to \sqrt{(\gamma^*)'}$. 
    
    Since continuous functions on the compact interval $I$ are bounded, the product of these two bounded, uniformly convergent sequences likewise converges uniformly:
    $$
        g(\gamma_{n_k}(t))\sqrt{\gamma'_{n_k}(t)} \to g(\gamma^*(t))\sqrt{(\gamma^*)'(t)} \quad \text{uniformly on } I
    $$

    Given $f \in C^0(I)$, the entire integrand converges uniformly on the compact interval $I$. This uniform convergence permits the limit to pass inside the integral:
    $$
        \lim_{k \to \infty} \mathcal{D}_3(\psi_{n_k}) = \mathcal{D}_3(\psi^*)
    $$
    
    Finally, utilizing the weak lower semi-continuity of the Sobolev norm $\|\psi\|_{\mathbb{H}}^2$, we conclude:
    $$
        \mathcal{O}_3(\psi^*) \leq \lim_{k \to \infty} \mathcal{D}_3(\psi_{n_k}) + \lambda \liminf_{k \to \infty} \|\psi_{n_k}\|_{\mathbb{H}}^2 = \inf_{\psi \in \mathbb{H}} \mathcal{O}_3(\psi)
    $$
    Thus, $\psi^*$ is a global minimizer for the isometry-based functional.
\end{proof}

\subsection{Existence Proof for Method 4: Jacobian-Weighted $\mathbb{L}^2$ Mismatch}
\label{app:existence_proof4}

\begin{proof}
    Since the mismatch term $\mathcal{D}_4(\psi) \ge 0$, the functional $\mathcal{O}_4(\psi)$ is bounded below, and the penalty term $\lambda \|\psi\|_{\mathbb{H}}^2$ establishes coercivity as $\|\psi\|_{\mathbb{H}} \to \infty$. Let $\{\psi_n\} \subset \mathbb{H}$ be a minimizing sequence. Due to this coercivity, the sequence is bounded in $\mathbb{H}$. By the reflexivity of Hilbert spaces, there exists a subsequence such that $\psi_{n_k} \rightharpoonup \psi^*$ weakly in $\mathbb{H}$.
    
    As established in the proof of Method 1, the compact embedding $\mathbb{H} \hookrightarrow C^1(I)$ ensures that the weak convergence in $\mathbb{H}$ implies strong convergence in the $C^1(I)$ norm. This yields uniform convergence for the warping map and its derivative: 
    $$
        \|\gamma_{n_k} - \gamma^*\|_\infty \to 0, \quad \|\gamma'_{n_k} - (\gamma^*)'\|_\infty \to 0
    $$

    Given $f, g \in C^0(I)$, the uniform convergence of $\gamma_{n_k}$ implies the squared residual $(f - g \circ \gamma_{n_k})^2 \to (f - g \circ \gamma^*)^2$ uniformly. Furthermore, since $(\gamma^*)'$ is continuous and strictly positive in our CLR framework, and the square root function is uniformly continuous on $[0, \infty)$, the weighting factor $\sqrt{\gamma'_{n_k}}$ is well-defined and converges uniformly to $\sqrt{(\gamma^*)'}$. 
    
    Because continuous functions on a compact interval are bounded, the product of these two bounded, uniformly convergent sequences also converges uniformly on $I$:
    $$
        (f(t) - g(\gamma_{n_k}(t)))^2 \sqrt{\gamma'_{n_k}(t)} \to (f(t) - g(\gamma^*(t)))^2 \sqrt{(\gamma^*)'(t)} \quad \text{uniformly}
    $$

    This uniform convergence on the compact interval $I$ allows us to pass the limit inside the integral directly:
    $$
        \lim_{k \to \infty} \mathcal{D}_4(\psi_{n_k}) = \mathcal{D}_4(\psi^*)
    $$
    
    By the weak lower semi-continuity of the Sobolev norm $\|\psi\|_{\mathbb{H}}^2$, we conclude:
    $$
        \mathcal{O}_4(\psi^*) \leq \lim_{k \to \infty} \mathcal{D}_4(\psi_{n_k}) + \lambda \liminf_{k \to \infty} \|\psi_{n_k}\|_{\mathbb{H}}^2 = \inf_{\psi \in \mathbb{H}} \mathcal{O}_4(\psi)
    $$
    This confirms that $\psi^*$ is a global minimizer for the Jacobian-weighted functional.
\end{proof}

\section{Detailed Proofs on Noise-Free Asymptotic Consistency}
\label{app:consistency_proof}

\subsection{Proof of Lemma \ref{lemma:exist} for Existence of the Finite-Dimensional Minimizer}
\label{app:consistency_proof0}

\begin{proof}
To prove the existence of a global minimizer, we must show that $J_d(\mathbf{c})$ is continuous and coercive on $\mathbb{R}^{d}$. 

\textit{1. Continuity:} The coefficients $\mathbf{c}$ map linearly to the centered log-derivative (CLR) representation $\psi_{\mathbf{c}}$. The mapping from $\psi_{\mathbf{c}}$ to the warping function $\gamma_{\mathbf{c}}$ via the normalized exponential integration $\gamma_{\mathbf{c}}(t) = \int_0^t \frac{\exp(\psi_{\mathbf{c}}(s))}{\int_0^1 \exp(\psi_{\mathbf{c}}(u)) du} ds$ is continuous. Because the target signal $g$ is assumed to be at least continuous, the composition $g \circ \gamma_{\mathbf{c}}$ is continuous. Therefore, the data mismatch functional $\mathcal{D}_i$ is continuous with respect to $\mathbf{c}$. The Sobolev penalty term $\lambda_d \|\psi_{\mathbf{c}}\|_{\mathbb{H}}^2$ is quadratic in $\mathbf{c}$ and trivially continuous. Thus, $J_d(\mathbf{c})$ is a continuous function.

\textit{2. Coercivity:} We can write the Sobolev penalty strictly in terms of the coefficients as $\|\psi_{\mathbf{c}}\|_{\mathbb{H}}^2 = \mathbf{c}^\top \mathbf{R} \mathbf{c}$, where $\mathbf{R}$ is the Sobolev stiffness matrix. Because the basis functions $\{\phi_j\}_{j=1}^{d}$ are linearly independent and span a proper subspace of the centered Sobolev space, $\mathbf{R}$ is strictly positive definite. Therefore, there exists a minimum eigenvalue $\mu_{\min} > 0$ such that:
$$
    \mathbf{c}^\top \mathbf{R} \mathbf{c} \ge \mu_{\min} \|\mathbf{c}\|_2^2
$$
Because the data mismatch term is a squared distance or a non-negative integral (depending on the chosen $\mathcal{D}_i$), we have $\mathcal{D}_i \ge 0$. Consequently, the total objective is bounded below by the penalty:
$$
    J_d(\mathbf{c}) \ge \lambda_d \mu_{\min} \|\mathbf{c}\|_2^2
$$
Since $\lambda_d > 0$ and $\mu_{\min} > 0$, it follows that $J_d(\mathbf{c}) \to \infty$ as $\|\mathbf{c}\|_2 \to \infty$. This establishes that $J_d$ is coercive.

Because $J_d(\mathbf{c})$ is a continuous and coercive function on $\mathbb{R}^{d}$, its sub-level sets are closed and bounded (compact). By the generalized Weierstrass Extreme Value Theorem, a continuous function on a non-empty compact set attains its global minimum. Therefore, the minimizer $\hat{\mathbf{c}}_d$ exists. 
\end{proof}

\subsection{Proof of Proposition \ref{prop:consis1} for Consistency in Method 1}
\label{app:consistency_proof1}

\begin{proof}
Since $f$ is a \textit{reasonable} function, the global minimum of the mismatch functional $\mathcal{D}_1(\psi)$ is uniquely attained at the true warping $\gamma_\psi = \gamma_0$. By the density of the finite-dimensional subspaces $\mathbb{H}_{d}$ in $\mathbb{H}$, for the true centered log-derivative (CLR) representation $\psi_0$, there exists an approximating sequence $\tilde{\psi}_d \in \mathbb{H}_{d}$ such that $\|\tilde{\psi}_d - \psi_0\|_{\mathbb{H}} \to 0$ as $d \to \infty$. To establish that the finite-dimensional estimator $\hat{\psi}_d$ (in Equation \eqref{eq:con1}) converges to $\psi_0$, we analyze the optimality of the estimator through the following steps:

\begin{enumerate}
    \item \textit{The Optimality Gap:} By definition, $\hat{\psi}_d$ minimizes the objective functional over $\mathbb{H}_{d}$. For the approximating sequence $\tilde{\psi}_d \in \mathbb{H}_{d}$, we have:
    $$
        \mathcal{D}_1(\hat{\psi}_d) + \lambda_d \|\hat{\psi}_d\|_{\mathbb{H}}^2 \leq \mathcal{D}_1(\tilde{\psi}_d) + \lambda_d \|\tilde{\psi}_d\|_{\mathbb{H}}^2
    $$
    As $d \to \infty$, the right-hand side (RHS) vanishes. Specifically, the strong convergence $\|\tilde{\psi}_d - \psi_0\|_{\mathbb{H}} \to 0$ implies $\mathcal{D}_1(\tilde{\psi}_d) \to \mathcal{D}_1(\psi_0) = 0$ by the continuity of the integral mismatch. Furthermore, the penalty term $\lambda_d \|\tilde{\psi}_d\|_{\mathbb{H}}^2 \to 0$ since $\lambda_d \to 0$ and the sequence remains bounded in norm.

    \item \textit{Bounding the Sequence:} From the optimality inequality, we isolate the penalty term:
    $$
        \|\hat{\psi}_d\|_{\mathbb{H}}^2 \leq \frac{\mathcal{D}_1(\tilde{\psi}_d)}{\lambda_d} + \|\tilde{\psi}_d\|_{\mathbb{H}}^2
    $$
    By choosing the smoothing parameter such that the penalty vanishes slower than the approximation error (i.e., $\mathcal{D}_1(\tilde{\psi}_d) = o(\lambda_d)$), the first term on the RHS converges to zero. Since $\|\tilde{\psi}_d\|_{\mathbb{H}}^2 \to \|\psi_0\|_{\mathbb{H}}^2$, the sequence $\{\hat{\psi}_d\}$ is bounded in the Hilbert space $\mathbb{H}$. By the Banach-Alaoglu theorem, there exists a subsequence $\{\hat{\psi}_{d_k}\}$ that converges weakly to some $\psi^* \in \mathbb{H}$.

    \item \textit{Identifiability and Weak Limit:} From Step 1, it is evident that $0 \leq \mathcal{D}_1(\hat{\psi}_d) \leq \text{RHS} \to 0$, hence $\mathcal{D}_1(\hat{\psi}_d) \to 0$. In our setting, the Sobolev space $\mathbb{H}$ embeds compactly into $C^1([0,1])$. This compact embedding upgrades the weak convergence $\hat{\psi}_{d_k} \rightharpoonup \psi^*$ to strong uniform convergence of the CLR fields and their derivatives. 
    Consequently, the warping functions converge uniformly: $\gamma_{\hat{\psi}_{d_k}} \to \gamma_{\psi^*}$. By the continuity of $\mathcal{D}_1$ under uniform convergence, we have $\mathcal{D}_1(\psi^*) = 0$. Since $f$ is \textit{reasonable}, the equation $\mathcal{D}_1(\psi) = 0$ has a unique root, identifying $\psi^* = \psi_0$. Since every weakly convergent subsequence must converge to this same unique limit, the entire sequence weakly converges: $\hat{\psi}_d \rightharpoonup \psi_0$.

    \item \textit{Strong Convergence:} Weak convergence $\hat{\psi}_d \rightharpoonup \psi_0$ implies $\|\psi_0\|_{\mathbb{H}} \leq \liminf_{d \to \infty} \|\hat{\psi}_d\|_{\mathbb{H}}$ by the lower semi-continuity of the norm. Conversely, Step 2 established that \\
    $\limsup_{d \to \infty} \|\hat{\psi}_d\|_{\mathbb{H}} \leq \|\psi_0\|_{\mathbb{H}}$. Thus, $\|\hat{\psi}_d\|_{\mathbb{H}} \to \|\psi_0\|_{\mathbb{H}}$. In a Hilbert space, weak convergence combined with convergence of the norm implies strong convergence (the Radon-Riesz property):
    $$
        \|\hat{\psi}_d - \psi_0\|_{\mathbb{H}} \to 0
    $$
\end{enumerate}
\end{proof}

\subsection{Proof of Proposition \ref{prop:consis2} for Consistency in Method 2}
\label{app:consistency_proof2}

\begin{proof}
By the density of the finite-dimensional subspaces $\mathbb{H}_{d}$ in $\mathbb{H}$, there exists a sequence $\tilde{\psi}_d \in \mathbb{H}_{d}$ such that $\|\tilde{\psi}_d - \psi_0\|_{\mathbb{H}} \to 0$ as $d \to \infty$. At the ground truth, $\mathcal{D}_2(\psi_0) = 0$ because both the forward alignment $f(t) = g(\gamma_0(t))$ and the inverse alignment $f(\gamma_0^{-1}(u)) = g(u)$ hold exactly.

\begin{enumerate}
    \item \textit{Optimality and Boundedness:} Following identical logic to Method 1, the minimality of $\hat{\psi}_d$ (in Equation \eqref{eq:con2}) in $\mathbb{H}_{d}$ yields the bound:
    $$
        \|\hat{\psi}_d\|_{\mathbb{H}}^2 \leq \frac{\mathcal{D}_2(\tilde{\psi}_d)}{\lambda_d} + \|\tilde{\psi}_d\|_{\mathbb{H}}^2
    $$
    Under the condition $\mathcal{D}_2(\tilde{\psi}_d) = o(\lambda_d)$, the first term on the right-hand side vanishes. This ensures $\{\hat{\psi}_d\}$ is a bounded sequence in $\mathbb{H}$. By the same Banach-Alaoglu and compact embedding arguments established in Method 1, we can extract a weakly convergent subsequence $\hat{\psi}_{d_k} \rightharpoonup \psi^*$ which converges uniformly in both the warping functions and their derivatives.

    \item \textit{Symmetric Identifiability:} From the optimality bound, $0 \leq \mathcal{D}_2(\hat{\psi}_d) \to 0$. By applying a change of variables $u = \gamma_\psi(t)$ to the $\gamma'_\psi(t)$ portion of the stretching factor, we can expand the $\mathcal{D}_2$ integral into two explicitly symmetric terms:
    $$
        \mathcal{D}_2(\psi) = \frac{1}{2} \int_{0}^{1} (f(t) - g(\gamma_\psi(t)))^2 dt + \frac{1}{2} \int_{0}^{1} (f(\gamma_\psi^{-1}(u)) - g(u))^2 du
    $$

    Since $\mathcal{D}_2(\hat{\psi}_d) \to 0$ and both integrals are strictly non-negative, the forward and inverse mismatch functionals must simultaneously vanish. Because $f$ and $g$ are \textit{reasonable} functions, the vanishing of either integral uniquely identifies the limit. This dual constraint algebraically reinforces that the sequence weakly converges to the true unique warping: $\hat{\psi}_d \rightharpoonup \psi_0$.

    \item \textit{Strong Convergence:} Just as in Method 1, the weak convergence in $\mathbb{H}$ combined with the norm control $\limsup_{d \to \infty} \|\hat{\psi}_d\|_{\mathbb{H}} \leq \|\psi_0\|_{\mathbb{H}}$ guarantees strong convergence via the Radon-Riesz property:
    $$
        \|\hat{\psi}_d - \psi_0\|_{\mathbb{H}} \to 0
    $$
\end{enumerate}
\end{proof}

\subsection{Proof of Proposition \ref{prop:consis4} for Consistency in Method 4}
\label{app:consistency_proof4}

\begin{proof}
By substituting the ground truth $f = g \circ \gamma_0$ into the mismatch, we observe that at $\psi = \psi_0$:
$$
    \mathcal{D}_4(\psi_0) = \int_{0}^{1} (g(\gamma_0(t)) - g(\gamma_0(t)))^2 \sqrt{\gamma'_0(t)} dt = 0.
$$
The weighting factor $\sqrt{\gamma'_0(t)}$ acts purely as a local volume measure on the residual, preserving the location of the global minimum.

\begin{enumerate}
    \item \textit{Boundedness:} Using the optimality of $\hat{\psi}_d$ (in Equation \eqref{eq:con4}) within $\mathbb{H}_{d}$ and the approximation sequence $\tilde{\psi}_d \to \psi_0$, we have:
    $$
        \|\hat{\psi}_d\|_{\mathbb{H}}^2 \leq \frac{\mathcal{D}_4(\tilde{\psi}_d)}{\lambda_d} + \|\tilde{\psi}_d\|_{\mathbb{H}}^2.
    $$
    The condition $\mathcal{D}_4(\tilde{\psi}_d) = o(\lambda_d)$ ensures the first term vanishes as $d \to \infty$. Thus, the sequence $\{\hat{\psi}_d\}$ is bounded in $\mathbb{H}$, allowing the extraction of a weakly convergent subsequence $\hat{\psi}_{d_k} \rightharpoonup \psi^*$.

    \item \textit{Identifiability:} The optimality bound forces $\mathcal{D}_4(\hat{\psi}_d) \to 0$. Because the Sobolev penalty ensures $\sqrt{\gamma'_{\hat{\psi}_d}(t)} > 0$ strictly for all $t$, the integral vanishes if and only if the signal difference $(f(t) - g(\gamma_{\hat{\psi}_d}(t)))^2$ vanishes almost everywhere. For reasonable functions, this uniquely identifies the limit: $\psi^* = \psi_0$.

    \item \textit{Norm Convergence:} Weak convergence in $\mathbb{H}$ and the asymptotic norm bound \\
    $\limsup_{d \to \infty} \|\hat{\psi}_d\|_{\mathbb{H}} \leq \|\psi_0\|_{\mathbb{H}}$ jointly yield convergence in the $\mathbb{H}$ norm via the Radon-Riesz property:
    $$
        \|\hat{\psi}_d - \psi_0\|_{\mathbb{H}} \to 0.
    $$
\end{enumerate}
\end{proof}

\section{Mathematical Details of Fr\'{e}chet Derivatives}
\label{app:derivatives}

This appendix provides the rigorous derivations for the Fr\'{e}chet derivatives of the objective functionals. To maintain consistency with the main text, we denote the distance terms as $\mathcal{D}_1$ (Standard $\mathbb{L}^2$), $\mathcal{D}_2$ (Symmetric $\mathbb{L}^2$), $\mathcal{D}_3$ (Isometry), and $\mathcal{D}_4$ (Jacobian-Weighted $\mathbb{L}^2$).

We consider the space of centered log-derivative fields $\psi \in \mathbb{L}^2(I)$. The warping function $\gamma: I \to I$ is generated via the normalized exponential map:
$$
    \gamma(t) = \frac{L(t)}{V}, \quad L(t) = \int_0^t \exp(\psi(s)) ds, \quad V = L(1)
$$
The Jacobian (velocity) of the warp is denoted by $\gamma'(t) = \frac{\exp(\psi(t))}{V}$. By construction, $\gamma(0)=0$, $\gamma(1)=1$, and $\gamma'(t) > 0$.

To compute the Fr\'{e}chet derivative of a functional $\mathcal{D}_i(\psi)$, we determine the G\^{a}teaux variation $\delta \gamma$ and $\delta \gamma'$ in the direction of a perturbation $h \in \mathbb{L}^2(I)$. Let $\bar{h}$ be the weighted expectation of the perturbation under the current density $\gamma'$:
$$
    \bar{h} = \int_0^1 h(s) \gamma'(s) ds
$$

The variation of the Jacobian is derived via the quotient rule applied to $\gamma' = \exp(\psi)/V$:
$$
    \delta \gamma'(t) = \gamma'(t) \left( h(t) - \bar{h} \right)
$$
Integrating $\delta \gamma'$ from $0$ to $t$ yields the variation of the position:
$$
    \delta \gamma(t) = \int_0^t \gamma'(s) h(s) ds - \gamma(t) \bar{h}
$$

\subsection{Fr\'{e}chet Derivative of the Standard $\mathbb{L}^2$ Metric ($\mathcal{D}_1$)}

The objective for the Standard $\mathbb{L}^2$ framework is $\mathcal{D}_1(\psi) = \int_0^1 (f(t) - g(\gamma(t)))^2 dt$. The first-order variation with respect to the warping function is:
$$
    \delta \mathcal{D}_1 = \int_0^1 \underbrace{-2(f(t) - g(\gamma(t))) g'(\gamma(t))}_{m_1(t)} \delta \gamma(t) dt
$$

Substituting the expansion for $\delta \gamma(t)$ derived above, we apply Fubini's Theorem to swap the order of integration for the double integral involving $\int_0^t \gamma'(s) h(s) ds$:
$$
    \delta \mathcal{D}_1 = \int_0^1 \gamma'(s) \left( \int_s^1 m_1(t) dt \right) h(s) ds - \left( \int_0^1 m_1(t) \gamma(t) dt \right) \bar{h}
$$

We define the adjoint kernel as the cumulative positional force: $\mathcal{K}_1(s) = \int_s^1 m_1(t) dt$. Noting via integration by parts that the second term evaluates exactly to the weighted expectation of the kernel $\bar{\mathcal{K}}_1 = \int_0^1 \mathcal{K}_1(u) \gamma'(u) du$, we can isolate $h(s)$ to identify the Fr\'{e}chet derivative $\frac{\delta \mathcal{D}_1}{\delta \psi}$:
$$
    \frac{\delta \mathcal{D}_1}{\delta \psi}(t) = \gamma'(t) \left( \mathcal{K}_1(t) - \bar{\mathcal{K}}_1 \right)
$$

\subsection{Fr\'{e}chet Derivative of the Symmetric $\mathbb{L}^2$ Metric ($\mathcal{D}_2$)}

The symmetric objective is expressed in the domain of $t$ as:
$$
    \mathcal{D}_2(\psi) = \frac{1}{2}\left[\int_{0}^{1} (f(t) - g(\gamma(t)))^2 dt + \int_{0}^{1} (g(\gamma(t)) - f(t))^2 \gamma'(t) dt\right]
$$

Applying the G\^{a}teaux variation to both terms, we obtain:
$$
    \delta \mathcal{D}_2 = \int_0^1 m_{pos}(t) \delta \gamma(t) dt + \int_0^1 m_{jac}(t) \delta \gamma'(t) dt
$$
where the positional and Jacobian residuals (absorbing the $\frac{1}{2}$ scaling factor from the objective) are defined as:
$$
\begin{aligned}
    m_{pos}(t) &= -(f(t) - g(\gamma(t)))g'(\gamma(t))(1 + \gamma'(t)) \\
    m_{jac}(t) &= \frac{1}{2}(g(\gamma(t)) - f(t))^2
\end{aligned}
$$

Using integration by parts to map the positional force to the velocity variation:
$$
    \int_0^1 m_{pos}(t) \delta \gamma(t) dt = \int_0^1 \left( \int_t^1 m_{pos}(s) ds \right) \delta \gamma'(t) dt
$$

Summing the contributions, the total variation becomes:
$$
    \delta \mathcal{D}_2 = \int_0^1 \underbrace{\left( \int_t^1 m_{pos}(s) ds + m_{jac}(t) \right)}_{\mathcal{K}_{2}(t)} \delta \gamma'(t) dt
$$

Substituting $\delta \gamma'(t) = \gamma'(t)(h(t) - \bar{h})$, we arrive at the final Fr\'{e}chet derivative:
$$
    \frac{\delta \mathcal{D}_2}{\delta \psi}(t) = \gamma'(t) \left( \mathcal{K}_{2}(t) - \bar{\mathcal{K}}_{2} \right)
$$
where $\bar{\mathcal{K}}_{2} = \int_0^1 \mathcal{K}_{2}(u) \gamma'(u) du$.

\subsection{Fr\'{e}chet Derivative of the Isometry ($\mathbb{L}^2$-Preserving) Metric ($\mathcal{D}_3$)}

The isometry objective preserves signal energy by treating functions as square-integrable half-densities:
$$
    \mathcal{D}_3(\psi) = \int_0^1 \left( f(t) - g(\gamma(t))\sqrt{\gamma'(t)} \right)^2 dt
$$

Applying the G\^{a}teaux variation $\gamma \to \gamma + \epsilon \delta \gamma$, we derive the variation $\delta \mathcal{D}_3$:
$$
    \delta \mathcal{D}_3 = \int_0^1 -2\left(f(t) - g(\gamma(t))\sqrt{\gamma'(t)}\right) \left[ g'(\gamma(t))\sqrt{\gamma'(t)} \delta \gamma(t) + \frac{g(\gamma(t))}{2\sqrt{\gamma'(t)}} \delta \gamma'(t) \right] dt
$$

Identifying the positional and Jacobian-based residuals:
$$
\begin{aligned}
    m_{pos}(t) &= -2\left(f(t) - g(\gamma(t))\sqrt{\gamma'(t)}\right) g'(\gamma(t))\sqrt{\gamma'(t)} \\
    m_{jac}(t) &= -\left(f(t) - g(\gamma(t))\sqrt{\gamma'(t)}\right) \frac{g(\gamma(t))}{\sqrt{\gamma'(t)}}
\end{aligned}
$$

Using integration by parts with the boundary conditions $\delta \gamma(0) = \delta \gamma(1) = 0$, the total adjoint kernel $\mathcal{K}_{3}(t)$ is constructed as:
$$
    \mathcal{K}_{3}(t) = \int_t^1 m_{pos}(s) ds + m_{jac}(t)
$$

The final Fr\'{e}chet derivative is then:
$$
    \frac{\delta \mathcal{D}_3}{\delta \psi}(t) = \gamma'(t) \left( \mathcal{K}_{3}(t) - \bar{\mathcal{K}}_{3} \right)
$$
where 
$
    \bar{\mathcal{K}}_{3} = \int_0^1 \mathcal{K}_{3}(u) \gamma'(u) du.
$

\subsection{Fr\'{e}chet Derivative of the Jacobian-Weighted $\mathbb{L}^2$ Metric ($\mathcal{D}_4$)}

The Jacobian-weighted objective incorporates the square root of the Jacobian as a weighting measure:
$$
    \mathcal{D}_4(\psi) = \int_0^1 \left( f(t) - g(\gamma(t)) \right)^2 \sqrt{\gamma'(t)} dt
$$

The G\^{a}teaux variation $\delta \mathcal{D}_4$ is:
$$
    \delta \mathcal{D}_4 = \int_0^1 \left[ -2(f(t) - g(\gamma(t)))g'(\gamma(t))\sqrt{\gamma'(t)} \right] \delta \gamma(t) dt + \int_0^1 \left[ \frac{(f(t) - g(\gamma(t)))^2}{2\sqrt{\gamma'(t)}} \right] \delta \gamma'(t) dt
$$

Identifying the positional and Jacobian residuals:
$$
\begin{aligned}
    m_{pos}(t) &= -2(f(t) - g(\gamma(t))) g'(\gamma(t))\sqrt{\gamma'(t)} \\
    m_{jac}(t) &= \frac{(f(t) - g(\gamma(t)))^2}{2\sqrt{\gamma'(t)}}
\end{aligned}
$$

The total adjoint kernel $\mathcal{K}_{4}(t)$ is constructed via the same integration-by-parts argument used in the previous mismatch terms:
$$
    \mathcal{K}_{4}(t) = \int_t^1 m_{pos}(s) ds + m_{jac}(t)
$$

The final Fr\'{e}chet derivative is:
$$
    \frac{\delta \mathcal{D}_4}{\delta \psi}(t) = \gamma'(t) \left( \mathcal{K}_{4}(t) - \bar{\mathcal{K}}_{4} \right)
$$
where $\bar{\mathcal{K}}_{4} = \int_0^1 \mathcal{K}_{4}(u) \gamma'(u) du$.

\section*{Declaration of Generative AI Use}

During the preparation of this work the author used Google Gemini in order to improve language and readability. After using this tool/service, the author reviewed and edited the content as needed and takes full responsibility for the content of the published article.

\bibliography{refs}

@article{ahn2020regression,
  title={Regression models using shapes of functions as predictors},
  author={Ahn, Kwang-Il and Tucker, J. Derek and Wu, Wei and Srivastava, Anuj},
  journal={Computational Statistics \& Data Analysis},
  volume={151},
  pages={107017},
  year={2020},
  publisher={Elsevier}
}

@article{eilers2004parametric,
  title={Parametric time warping},
  author={Eilers, P.H.},
  journal={Analytical Chemistry},
  volume={76},
  number={2},
  pages={404--411},
  year={2004}
}

@article{wang1997alignment,
  title={Alignment of curves by dynamic time warping},
  author={Wang, Kongming and Gasser, Theo},
  journal={The Annals of Statistics},
  volume={25},
  number={3},
  pages={1251--1276},
  year={1997},
  publisher={Institute of Mathematical Statistics}
}

@article{tagare2009symmetric_theory,
  title={Symmetric Non-rigid Registration: {A} Geometric Theory and Some Numerical Techniques},
  author={Tagare, Hemant D. and Groisser, David and Skrinjar, Oskar},
  journal={Journal of Mathematical Imaging and Vision},
  volume={34},
  number={1},
  pages={61--88},
  year={2009},
  publisher={Springer}
}

@article{tang2008pairwise,
  title={Pairwise curve synchronization for functional data},
  author={Tang, Ruoyong and M{\"u}ller, Hans-Georg},
  journal={Biometrika},
  volume={95},
  number={4},
  pages={875--889},
  year={2008},
  publisher={Oxford University Press}
}

@book{srivastava2016functional,
  title={Functional and shape data analysis},
  author={Srivastava, Anuj and Klassen, Eric P},
  year={2016},
  publisher={Springer}
}

@article{kneip2008combining,
  title={Combining registration and fitting for functional models},
  author={Kneip, Alois and Ramsay, James O},
  journal={Journal of the American Statistical Association},
  volume={103},
  number={483},
  pages={1155--1165},
  year={2008},
  publisher={Taylor \& Francis}
}

@article{telesca2008bayesian,
	title={Bayesian hierarchical curve registration},
	author={Telesca, Donatello and Inoue, Lurdes Y T},
	journal={Journal of the American Statistical Association},
	volume={103},
	number={481},
	pages={328--339},
	year={2008},
	publisher={Taylor \& Francis}
}

@article{Bharath2020landmark,
author = {Karthik Bharath and Sebastian Kurtek},
title = {Distribution on Warp Maps for Alignment of Open and Closed Curves},
journal = {Journal of the American Statistical Association},
volume = {115},
number = {531},
pages = {1378--1392},
year = {2020},
publisher = {Taylor \& Francis}
}

@article{gervini2004self,
  title={Self-modelling warping functions},
  author={Gervini, D. and Gasser, T.},
  journal={Journal of the Royal Statistical Society: Series B (Statistical Methodology)},
  volume={66},
  number={4},
  pages={959--971},
  year={2004}
}

@article{james2007curve,
  title={Curve alignment by moments},
  author={James, G.M.},
  journal={The Annals of Applied Statistics},
  volume={1},
  number={2},
  pages={480--501},
  year={2007}
}

@article{egozcue2006hilbert,
  title={Hilbert space of probability density functions based on Aitchison geometry},
  author={Egozcue, Juan Jos{\'e} and D{\'\i}az-Barrero, Jos{\'e} Luis and Pawlowsky-Glahn, Vera},
  journal={Acta Mathematica Sinica, English Series},
  volume={22},
  pages={1175--1182},
  year={2006},
  publisher={Springer}
}

@article{marron2015functional,
  title={Functional Data Analysis of Amplitude and Phase Variation},
  author={Marron, J. S. and Ramsay, James O. and Sangalli, Laura M. and Srivastava, Anuj},
  journal={Statistical Science},
  volume={30},
  number={4},
  pages={468--484},
  year={2015},
  publisher={Institute of Mathematical Statistics}
}

@article{tucker2013generative,
  title={Generative models for functional data using phase and amplitude separation},
  author={Tucker, J. Derek and Wu, Wei and Srivastava, Anuj},
  journal={Computational Statistics \& Data Analysis},
  volume={61},
  pages={50--66},
  year={2013},
  publisher={Elsevier}
}

@article{lee2016fpca,
  title={Combined analysis of amplitude and phase variations in functional data},
  author={Lee, Sungwon and Jung, Sungkyu},
  journal={arXiv preprint arXiv:1603.01775},
  year={2016}
}

@misc{fsdd_github,
  author = {Jackson, Z. and Souza, C. and Flaks, J. and Pan, Y. and Nicolas, H. and Thite, A.},
  title = {Free Spoken Digit Dataset (FSDD)},
  year = {2018},
  publisher = {GitHub},
  journal = {GitHub repository},
  howpublished = {\url{https://github.com/Jakobovski/free-spoken-digit-dataset}},
  note = {Accessed: 2026-02-18}
}

@book{adams2003sobolev,
  title={Sobolev Spaces},
  author={Adams, Robert A. and Fournier, John J. F.},
  volume={140},
  year={2003},
  publisher={Elsevier Science},
  edition={2nd},
  note={Foundational reference for the Sobolev Embedding Theorem and $L^\infty$ bounds in one dimension.}
}

@book{deBoor2001,
  title={A Practical Guide to Splines},
  author={de Boor, Carl},
  year={2001},
  publisher={Springer-Verlag},
  address={New York},
  note={Standard reference for using B-splines as a basis for sieve estimation in functional spaces.}
}

@article{matuk2021bayesian,
  title={Bayesian framework for simultaneous registration and estimation of noisy, sparse, and fragmented functional data},
  author={Matuk, James and Bharath, Karthik and Chkrebtii, Oksana and Kurtek, Sebastian},
  journal={Journal of the American Statistical Association},
  volume={116},
  number={534},
  pages={1--17},
  year={2021}
}

@article{tucker2021multimodal,
  title={Multimodal Bayesian registration of noisy functions using Hamiltonian Monte Carlo},
  author={Tucker, J Derek and Shand, Lyndsay and Chowdhary, K},
  journal={Computational Statistics \& Data Analysis},
  volume={163},
  pages={107298},
  year={2021}
}

@article{cheng2016bayesian,
  title={Bayesian registration of functions and curves},
  author={Cheng, Wei and Dryden, Ian L and Huang, Xiaolei},
  journal={Bayesian Analysis},
  volume={11},
  pages={447--475},
  year={2016}
}

@article{kurtek2017geometric,
  title={A geometric approach to pairwise Bayesian alignment of functional data using importance sampling},
  author={Kurtek, Sebastian},
  journal={Electronic Journal of Statistics},
  volume={11},
  pages={502--531},
  year={2017}
}

@article{lu2017bayesian,
  title={Bayesian registration of functions with a Gaussian process prior},
  author={Lu, Y and Herbei, R and Kurtek, S},
  journal={Journal of Computational and Graphical Statistics},
  volume={26},
  pages={894--904},
  year={2017}
}

@article{ma2025new,
  title={A new framework for Bayesian function registration},
  author={Ma, Yijia and Wu, Wei},
  journal={Electronic Journal of Statistics},
  volume={19},
  pages={2171--2198},
  year={2025},
  doi={10.1214/25-EJS2383}
}

@article{ramsay1998curve,
  title     = {Curve registration by nonparametric hierarchical modeling},
  author    = {Ramsay, J. O. and Li, X.},
  journal   = {Journal of the Royal Statistical Society Series B: Statistical Methodology},
  volume    = {60},
  number    = {2},
  pages     = {351--363},
  year      = {1998},
  publisher = {Oxford University Press}
}

@book{ramsay2005functional,
  title     = {Functional Data Analysis},
  author    = {Ramsay, J. O. and Silverman, B. W.},
  year      = {2005},
  publisher = {Springer},
  address   = {New York, NY},
  edition   = {2nd},
  series    = {Springer Series in Statistics},
  doi       = {10.1007/b98888}
}

@article{srivastava2011,
  author  = {Srivastava, A. and Wu, W. and Kurtek, S. and Klassen, E. and Marron, J. S.},
  title   = {Registration of functional data using Fisher-Rao metric},
  journal = {arXiv preprint arXiv:1103.3817},
  year    = {2011}
}

@article{ma2024,
  title     = {A stochastic process representation for time warping functions},
  author    = {Ma, Y. and Zhou, X. and Wu, W.},
  journal   = {Computational Statistics \& Data Analysis},
  volume    = {194},
  pages     = {107941},
  year      = {2024},
  publisher = {Elsevier},
  doi       = {https://doi.org/10.1016/j.csda.2024.107941}
}
\bibliographystyle{abbrvnat}

\end{document}